\documentclass[letterpaper,11pt]{article}

\usepackage{graphicx} % Required for inserting images
\usepackage[margin=1in]{geometry}
\usepackage{algorithm}
\usepackage[noend]{algpseudocode}
\usepackage{amsthm}
\usepackage{amssymb,amsmath}  
\usepackage{thmtools}  
\usepackage{mathtools}
\usepackage[colorlinks=true, allcolors=blue]{hyperref}
\usepackage[capitalize]{cleveref}
\usepackage{todonotes}
\usepackage{placeins}
\usepackage{subcaption}
\usepackage{booktabs}

\usepackage{newtxtext}
\usepackage{newtxmath}

\usepackage{enumitem}
\setlist[1]{labelindent=\parindent}
\setlist[enumerate]{label=(\arabic*)}
\setlist[itemize]{noitemsep}

\usepackage{authblk}

\author[1,2]{Ben Bals}
\author[3]{Panagiotis Charalampopoulos}
\author[4]{Oded Lachish}
\author[5]{Solon P. Pissis}
\author[1]{Hilde Verbeek}
\affil[1]{CWI, Amsterdam, The Netherlands}
\affil[2]{Vrije Universiteit, Amsterdam, The Netherlands}
\affil[3]{King's College London, UK}
\affil[4]{Birkbeck, University of London, UK}
\affil[5]{The Cyprus Institute, Nicosia, Cyprus}
\date{\today}

\declaretheorem[numberwithin=section]{theorem}

\declaretheorem[sibling=theorem]{corollary}
\declaretheorem[sibling=theorem]{lemma}
\declaretheorem[sibling=theorem]{definition}

\declaretheorem[sibling=theorem]{claim}
\declaretheorem[sibling=theorem]{remark}
\declaretheorem[sibling=theorem]{example}

\declaretheorem[sibling=theorem]{proposition}

\makeatletter
\newcommand{\proofsubparagraph}[1]{\noindent{\color{darkgray}\normalsize\bfseries #1}}

\newenvironment{claimproof}[1][\proofname]{
  \pushQED{\qed}%
  \normalfont \topsep6\p@\@plus6\p@\relax
  \trivlist
  \item[\hskip\labelsep
        \color{darkgray}
    #1\@addpunct{.}]\ignorespaces
}{%
  
  \popQED\endtrivlist%\@endpefalse
  
}
\makeatother

\usepackage{mathtools}
\usepackage{xspace}
\usepackage{tikz}
\usetikzlibrary{decorations.pathreplacing}

\newcommand{\defDSproblem}[3]{
\vspace{2mm}
\noindent\fbox{
   \begin{minipage}{0.9\columnwidth}
   \textsc{#1}\\
   {\bf{Preprocess:}} #2  \\
   {\bf{Query:}} #3
   \end{minipage}
   }
   \vspace{2mm}
}

\def\dd{\mathinner{.\,.}}
\newcommand{\polylog}{\mathrm{polylog}}

\newcommand{\str}{\textsf{str}}
\newcommand{\per}{\textsf{per}}
\newcommand{\cO}{\mathcal{O}}
\newcommand{\cOt}{\widetilde{\mathcal{O}}}

\newcommand{\Stat}{\textsf{Stat}}
\newcommand{\aggOp}{\bigodot}

\newcommand{\Occ}{\textsf{Occ}}
\newcommand{\Occno}{\Occ^{\text{NO}}}
\newcommand{\Occco}{\Occ^{\alpha\beta}}
\newcommand{\HO}{\textsf{H$z$O}}

\newcommand{\AOP}{\textsf{A$z$OP}}
\newcommand{\USI}{USI\xspace}

\newcommand{\COC}{COC\xspace}
\newcommand{\setsize}[1]{\left|#1\right|}
\newcommand{\E}{\mathbb{E}}

\newcommand{\ie}[1]{(i.e.,~#1)}
\newcommand{\eg}[1]{(e.g.,~#1)}

\newcommand{\DrawRun}[4]{
	\def\h{0.5};
	\def\ext{0.4};

    \filldraw [fill=#4] (#1, #2)
		-- +(-\ext, 0)
		-- +(-\ext - 0.1, 0.5 * \h)
        -- +(-\ext, \h)
		-- +(#3 + \ext, \h)
		-- +(#3 + \ext + 0.1, 0.5 * \h)
        -- +(#3 + \ext, 0)
		-- +(0, 0);

    \foreach \i in {0,...,#3} {
        \draw (#1 + \i, #2) -- +(0, \h);
    }
}

\newcommand{\DrawRunExt}[5]{
	\def\h{0.5};
	\def\ext{0.4};

    \filldraw [fill=#5] (#1, #2)
		-- +(-\ext, 0)
		-- +(-\ext - 0.1, 0.5 * \h)
        -- +(-\ext, \h)
        -- +(0, \h);

    \filldraw [fill=#5] (#1 + #3, #2)
        -- +(\ext, 0)
        -- +(\ext + 0.1, 0.5 * \h)
        -- +(\ext, \h)
        -- +(0, \h);

    \filldraw [fill=#4] (#1, #2) -- +(#3, 0) -- +(#3, \h) -- +(0, \h);

    \foreach \i in {0,...,#3} {
        \draw (#1 + \i, #2) -- +(0, \h);
    }
}

\title{Text Indexing: From Reporting to Counting}
\date{\vspace{-.5cm}}

\begin{document}

\maketitle

\begin{abstract}
We prove an elementary yet powerful combinatorial lemma: 
 in any rooted tree with $L$ leaves, the number of nodes whose depth is smaller than the number of their leaf descendants is at most~$L$.
 For any string $T$ of length $n$, a direct application of this lemma to the suffix trie of $T$ yields that
 the number of substrings of $T$ whose length is smaller than their number of occurrences in $T$ is at most~$n$.
 This combinatorial insight leads to space-efficient data structures with optimal query times for string \emph{counting} problems via the following algorithmic framework:
 store the counts for the at most~$n$ ``frequent'' substrings of $T$ in a preprocessing step, and use a \emph{reporting} query to count for the ``infrequent'' substrings.
 Our framework acts as a convenient black box, lifting indexes with reporting time $\cO(\setsize{P}+|\Occ_T(P)|)$ to support counting queries in time $\cO(\setsize{P})$, where $P$ is the queried pattern and $\Occ_T(P)$ is the set of occurrences of $P$ in~$T$.
 As applications, we show efficient indexes for
 consecutive occurrences, 
 weighted sequences, 
 strings with utilities, and
 non-overlapping occurrences.
\end{abstract}

\thispagestyle{empty}
\clearpage
\pagenumbering{arabic}
\setcounter{page}{1}

\section{Introduction}

Text indexing is a classical problem in computer science with many applications,
especially in bioinformatics~\cite{DBLP:books/cu/Gusfield1997} and information retrieval~\cite{DBLP:books/aw/Baeza-YatesR2011}. The goal of text indexing is to preprocess a  string $T=T[1\dd n]$ of length $n$, called the \emph{text}, into a data structure, called
the \emph{index}, supporting queries for a string $P=P[1\dd m]$ of length $m \leq n$,
called the \emph{pattern}. The most typical query types in text indexes are
\emph{reporting} and \emph{counting}: the former asks for the set $\Occ_T(P)$ of
all starting positions (occurrences) of $P$ in $T$, while the latter asks for $|\Occ_T(P)|$.

The study of text indexing dates back to the early 1970s,
when Weiner~\cite{DBLP:conf/focs/Weiner73} introduced suffix trees.
This foundational data structure occupies $\cO(n)$ space, can be constructed in $\cO(n)$ time~\cite{DBLP:conf/focs/Farach97}, and answers reporting queries in $\cO(m + |\Occ_T(P)|)$ time and counting queries in $\cO(m)$ time. Subsequent work on text indexing focused mainly on reducing the space
complexity---the main bottleneck in practice---often at the cost of increased
query time. Key milestones in this direction include the suffix
array~\cite{DBLP:books/daglib/GonnetBS92,DBLP:journals/siamcomp/ManberM93},
the compressed suffix array~\cite{DBLP:journals/siamcomp/GrossiV05}, the
FM-index~\cite{DBLP:journals/jacm/FerraginaM05}, as well as the
$r$-index~\cite{DBLP:conf/soda/GagieNP18} and other indexes for 
highly repetitive texts~\cite{DBLP:journals/csur/Navarro21}.

The \emph{suffix tree} of $T$ is the compacted trie of all the suffixes of $T$.
Starting from the (uncompacted) suffix trie of~$T$, the suffix tree is obtained
by suppressing all internal nodes with exactly one child: these nodes become
\emph{implicit}, corresponding to points on edges of the suffix tree rather than
explicit nodes. The remaining nodes---the root, all branching internal nodes, and
all leaves---are called \emph{explicit}. This compactification leads to $\cO(n)$ explicit nodes, for a text $T$ of length $n$, and therefore to an $\cO(n)$-space index.

Despite the considerable attention given to suffix trees, the following intuition appears to have gone unnoticed. The structural property that allows suffix tries to be compacted from up to $\Theta(n^2)$ to $\Theta(n)$ nodes---namely, that many nodes of the trie share the exact same set of leaf descendants along unary paths---also bounds the number of nodes whose depth is smaller than their number of leaf descendants. Our central contribution is an elementary yet powerful combinatorial lemma formalizing this intuition: in any rooted tree on $L$ leaves, the number of nodes whose depth is smaller than the number of their leaf descendants is at most~$L$. Applied to the suffix trie of a given string, this yields the following interpretation.

\begin{restatable}{theorem}{stringthm}
\label{thm:string}
    In any string $T$ of length $n$, the number of non-empty substrings of $T$ whose length is smaller than their number of occurrences in $T$ is at most $n$.
\end{restatable}

We prove that this bound is asymptotically tight for $T = (a_1 a_2 \dots a_r)^r$ over alphabet $\Sigma = \{a_1, a_2, \dots, a_r\}$.
The algorithmic implications of \Cref{thm:string} are far-reaching. \Cref{thm:string} leads to space-efficient data structures with
\emph{optimal query times} for string \emph{counting} problems via the following framework: store the counts for the at most~$n$ ``frequent'' substrings of $T$ in a preprocessing step, and use a \emph{reporting} query to count for the ``infrequent'' substrings.
Our framework acts as a convenient black box,
lifting indexes with reporting time $\cO(m+|\Occ_T(P)|)$ to counting indexes with query time $\cO(m)$. More formally, we prove the following theorem.

\begin{restatable}[Algorithmic Framework]{theorem}{thmDSBlackbox}
\label{thm:ds-blackbox}
    Let $D$ be a data structure constructed over a text $T$ of length~$n$, so 
    that, when queried for a pattern $P$ of length $m$, it returns an answer of total size $w$ in ${\cO(m+\setsize{\Occ_T(P)} + w)}$ time. Then, there is a data structure $\hat{D}$ of size $\cO(\setsize{D}+w\cdot n)$ that answers the same queries for $P$ in time $\cO(m + w)$.
\end{restatable}

Our framework can be easily extended for efficiently addressing the problem where each occurrence of a pattern has some associated value, and we wish to aggregate these values for all sought occurrences under an associative operator.
See~\cref{cor:ds-stat} for a formal statement.

As a first application, we consider the \emph{consecutive occurrences indexing} problem~\cite{DBLP:journals/algorithmica/BilleG14,
DBLP:journals/tcs/BilleGPRS22,
DBLP:journals/algorithmica/BilleGPS23,
DBLP:journals/mst/BilleGVV14,
DBLP:journals/is/GawrychowskiGSS26,
DBLP:journals/algorithmica/IliopoulosR09,
DBLP:journals/tcs/NavarroT16,
DBLP:conf/esa/Radoszewski23}.
In its most basic variant, we are asked to preprocess $T$ and an interval $[\alpha,\beta]$ into an index, such that, for any pattern $P$, we can report the consecutive occurrences of $P$ in $T$ with a distance in $[\alpha,\beta]$.
While a more general $\cO(n\log n)$-space index (that also gets the interval $[\alpha,\beta]$ at query time)
with optimal reporting queries exists due to Navarro and Thankachan~\cite{DBLP:journals/tcs/NavarroT16}, the counting version with fixed bounded-gap constraints has no known $\cO(n \cdot \polylog(n))$-space solution with $\cO(m)$-time queries.
A black-box application of \Cref{thm:ds-blackbox} to the main result of~\cite{DBLP:journals/tcs/NavarroT16} yields an $\cO(n\log n)$-space index with optimal $\cO(m)$-time counting queries. We improve the space complexity to $\cO(n)$.
We also present other implications of \Cref{thm:ds-blackbox} in data structures for indexing problems, including problems involving weighted sequences~\cite{DBLP:journals/iandc/BartonK0PR20,DBLP:journals/jea/Charalampopoulos20,DBLP:conf/icde/Gabory0LPZ24} and strings with utilities~\cite{DBLP:conf/sdm/0001CCGGLPP24,DBLP:conf/icde/BernardiniCCGGLPP25}.

For a class of indexing problems in the framework of \Cref{cor:ds-stat}, including  indexing strings with utilities, we further show that we can \emph{construct} an $\cO(n)$-space index  in $\cOt(n\sqrt{n})$ time by using \Cref{thm:string} and exploiting the periodic structure of $T$.
Our data structure improves over the state-of-the-art index for strings with utilities~\cite{DBLP:conf/icde/BernardiniCCGGLPP25} \emph{polynomially} in both the preprocessing time (from $\cO(n^2)$ to $\cOt(n\sqrt{n})$) as well as the space--query time trade-off (from $\cO(n^2/\tau)$ space and $\cO(m+\tau)$ query time to $\cO(n)$ space and $\cO(m)$ query time).

As a bonus, we unify the treatment of several indexing problems, recovering known results for their counting variants by applying our algorithmic framework to their reporting variants; we do so for indexing strings with properties~\cite{DBLP:journals/tcs/AmirCIKZ08,DBLP:journals/iandc/BartonK0PR20}, for indexing for overlapping or non-overlapping occurrences, and for covers~\cite{DBLP:conf/icalp/BrodalLOP02,DBLP:conf/isaac/CohenP09,DBLP:conf/esa/Radoszewski23}.

Notably, our combinatorial lemma on trees can be generalized to yield the following two corollaries on strings, which have further consequences.

\begin{restatable}{corollary}{stringcor}\label{cor:general}
    In any string $T$ of length $n$, for any integer $
    \tau \in [1,n]$, the number of non-empty substrings~$U$ of $T$ for which $\setsize{U} \cdot \tau < |\Occ_T(U)|$ is at most $n/\tau$.  
\end{restatable}

\begin{restatable}{corollary}{inflatecor}\label{cor:inflate}
    In any string $T$ of length $n$, for any integer $
    \tau \in [1,n]$, the number of non-empty substrings~$U$ of $T$ for which $\setsize{U} < \tau \cdot |\Occ_T(U)|$ is at most $n\tau$.
\end{restatable}  

\Cref{cor:general} enables us to transform state-of-the-art reporting data structures (e.g., \cite{DBLP:journals/algorithmica/GangulyST20,DBLP:journals/tcs/GibneyMT25}) that occupy less than $\cO(n)$ space (say, they use $o(n\log n)$ bits) to support counting queries.
Roughly speaking, given an $s$-space data structure, we set $\tau = n/s$, precompute and store the answers for at most $n/\tau = s$ substrings $U$ of $T$ (the ones for which $\setsize{U} \cdot \tau < \setsize{\Occ_T(U)}$), and spend $\cO(m \cdot n/s)$ time at query time for any given length-$m$ pattern. We use this idea to construct a compact data structure for counting non-overlapping occurrences in near-optimal time, which was mentioned in the literature (cf.~\cite{DBLP:journals/tcs/GibneyMT25}) as a ``challenging'' open problem.

Furthermore, \Cref{cor:inflate} enables us to extend the application of our framework to the \emph{internal pattern matching} (IPM) model~\cite{DBLP:journals/algorithmica/Charalampopoulos21,DBLP:journals/siamcomp/KociumakaRRW24}, where the pattern is specified as a fragment of $T$. We use this idea to construct an IPM data structure for indexing strings with utilities.

\paragraph{Paper organization.} In \Cref{sec:lemma}, we prove the combinatorial lemma and \Cref{thm:string}.
In \Cref{sec:framework}, we prove \Cref{thm:ds-blackbox} and present the algorithmic framework.
In \Cref{sec:headline}, we present the headline applications of our framework and demonstrate the versatility of our approach by recovering several known results under our framework. In \Cref{sec:general}, we prove \Cref{cor:general} and showcase an application.
In \Cref{sec:internal}, we prove \Cref{cor:inflate} and showcase the IPM index for strings with utilities as an application.
In \Cref{sec:construct}, we present a construction algorithm for a class of indexing problems, including indexing strings with utilities. 

\section{The Combinatorial Lemma}\label{sec:lemma}

\begin{lemma}\label{lem:tree}
For any rooted tree $\mathcal{T}$ with $L$ leaves over a set $\mathcal{V}$ of nodes, we have
\[
|\{v\in \mathcal{V} \mid \textsf{depth}(v) < |\mathcal{L}(v)|\}|\leq L,
\]
where $\textsf{depth}(v)$ is the depth of node $v$ and $\mathcal{L}(v)$ is the set of leaf descendants of $v$.
\end{lemma}

\begin{proof}
Let $\mathcal{L}(\mathcal{T})$ be the set of leaves of $\mathcal{T}$.
Further, let
$\mathcal{S} = \{v \in \mathcal{V} \mid \textsf{depth}(v) < |\mathcal{L}(v)|\}$.
We show that $|\mathcal{S}| \le L$ by constructing an \emph{injective} mapping $\phi: \mathcal{S} \to \mathcal{L}(\mathcal{T})$.

For any node $v \in \mathcal{S}$, the condition $\textsf{depth}(v) < |\mathcal{L}(v)|$ means that the number of leaf descendants of~$v$ is strictly greater than the depth of $v$.
We say that a node $v$ has in-order rank $i$ if it is the $i$th node visited in an in-order traversal of~$\mathcal{T}$.
Let the leaf descendants of $v$ be $\ell_{v,1}, \ell_{v,2}, \dots, \ell_{v,|\mathcal{L}(v)|}$, in increasing order with respect to their in-order ranks.
The function $\phi: \mathcal{S} \to \mathcal{L}(\mathcal{T})$ maps each node $v \in \mathcal{S}$ to its $(\textsf{depth}(v) + 1)$th leaf descendant according to the above order; i.e.,
$\phi(v) = \ell_{v, \textsf{depth}(v) + 1}$.
Since $v \in \mathcal{S}$, we have $|\mathcal{L}(v)| \ge \textsf{depth}(v) + 1$, ensuring that $\ell_{v, \textsf{depth}(v) + 1}$ exists within the subtree rooted at $v$. We must show that for any $u, v \in \mathcal{S}$ with $u \neq v$, we have $\phi(u) \neq \phi(v)$. We consider the following two cases:

\proofsubparagraph{\boldmath{Case 1: Neither $u$ nor $v$ is an ancestor of the other.}}
In this case, the subtrees rooted at $u$ and $v$ are disjoint. Consequently, their sets of leaf descendants $\mathcal{L}(u)$ and $\mathcal{L}(v)$ are disjoint.
Since $\phi(u) \in \mathcal{L}(u)$ and $\phi(v) \in \mathcal{L}(v)$, it immediately follows that $\phi(u) \neq \phi(v)$.

\proofsubparagraph{\boldmath{Case 2: $u$ is an ancestor of $v$ or vice versa.}}
We assume, without loss of generality, that $u$ is an ancestor of~$v$---the other case is symmetric.
Let $d_u \coloneqq \textsf{depth}(u)$ and $d_v \coloneqq \textsf{depth}(v)$. Since~\(u\) is a strict ancestor of~\(v\), we have \(d_u < d_v\). The leaf $\phi(u)$ is the $(d_u+1)$th leaf descendant of~$u$.
The leaf $\phi(v)$ is the $(d_v+1)$th leaf descendant of $v$.
Because \(v\) is in the subtree of \(u\), the leaf descendants of \(v\) form a contiguous segment within the leaf descendants of \(u\) (ordered by in-order ranks).
Let \(k \ge 0\) be the number of leaf descendants of \(u\) with smaller in-order rank than \(v\).
The index of \(\phi(v)\) among the leaf descendants of \(u\) is \(k + d_v + 1\).
Since \(k \ge 0\) and \(d_u < d_v\), we must have \(d_u + 1 < d_v + 1 \le k + d_v + 1\Rightarrow \phi(u) \neq \phi(v)\).

Since the mapping $\phi: \mathcal{S} \to \mathcal{L}(\mathcal{T})$ is injective, the cardinality of the set $\mathcal{S}$ must be less than or equal to the cardinality of the set $\mathcal{L}(\mathcal{T})$, and hence
$|\mathcal{S}| \le |\mathcal{L}(\mathcal{T})| = L$.
\end{proof}

The following consequence of the above lemma for strings is immediate.

\stringthm*
\begin{proof}
    Apply \Cref{lem:tree} to the trie of $\{T[1\dd n]\$, T[2\dd n]\$, \ldots, T[n]\$\}$, for a unique letter $\$$ that does not occur in~$T$.
    Since, by construction, there is a bijective mapping between the internal nodes of the trie and the substrings of $T$, the claim follows.
\end{proof}

We next show that the bounds in both \cref{lem:tree,thm:string} are asymptotically tight.

\begin{proposition}
\label{prop:string-lower-bound}
For every integer $r\geq 2$, there exists a string $T$ of length $n=r^2$ over an alphabet of size $r$ that has $n - 2\sqrt{n} + 2$ non-empty substrings $U$ satisfying $\setsize{U} < |\Occ_T(U)|$. 
\end{proposition}
\begin{proof}
    Let $T = X^r$, with $X = a_1 a_2 \dots a_r$ over $\Sigma = \{a_1, a_2, \dots, a_r\}$.
    Any substring $U$ of $T$ of length $\ell \le r$ is uniquely determined by its length and its starting position $i$ modulo $r$; such a substring $U$ occurs $r$ times if $(i-1) \bmod r + \ell - 1 \le r$, and $r-1$ times otherwise.
    For any $\ell \le r-2$, all $r$ length-$\ell$ substrings of $XX[1 \dd \ell -1]$ are distinct and have at least $r-1 > \ell$ occurrences in $T$.
    The number of these substrings is $r(r-2)$. 
    Further, both $X[1 \dd |X|-1]$ and $X[2 \dd |X|]$ have $r$ occurrences and length $r-1$, while all other substrings of $T$ of length $r-1$ have $r-1$ occurrences.
    Finally, note that any substring $U$ of length $\ell \ge r$ has at most $r$ occurrences, and thus $\ell \geq |\Occ_T(U)|$. 
    We conclude that the total number of substrings $U$ of $T$ that satisfy $\setsize{U} < |\Occ_T(U)|$ is $r(r-2) + 2 = n - 2\sqrt{n} + 2$.
\end{proof}

\section{Algorithmic Framework}\label{sec:framework}
\Cref{thm:string} enables efficient solutions for a number of interesting data structure problems on strings.
\Cref{thm:ds-blackbox} captures a wide class of such data structure problems where we obtain improvements in a black-box manner.

Below, we use perfect hashing to efficiently navigate the suffix tree.

\begin{lemma}[\cite{DBLP:conf/icalp/Ruzic08}]\label{lem:dict}
A static $\cO(n)$-space dictionary on a set of $n$ keys can be
deterministically constructed in time $\cO(n\log^2 \log n)$, so
that look-up queries to the dictionary take time $\cO(1)$.
\end{lemma}

\thmDSBlackbox*

There are many text indexing problems for which there is a data structure $D$ that reports some specific subset of $\Occ_T(P)$ in $\cO(m+|\mathrm{output}|)$ time.
For any such problem, the counting variant (that is, when we want to find the size of this subset) naturally satisfies the requirements of \Cref{thm:ds-blackbox}.
See \Cref{sec:consecutive,sec:weighted-sequences,sec:other-apps} for applications.

\begin{proof}
  We construct $\hat{D}$ as follows. We first perform the preprocessing of $D$ and then construct the suffix tree of $T$.
  The suffix tree takes $\cO(n)$ space~\cite{DBLP:conf/focs/Weiner73}.
  Each branching node of the suffix tree is
  augmented with a dictionary that supports $\cO(1)$-time access of the edges based on the first letter (key) of their label using \Cref{lem:dict}.
  The total size of the dictionaries is $\cO(n)$.
  
  In the suffix tree, we find the implicit or explicit nodes that correspond to ``frequent'' substrings \ie{those that occur more often than they are long} as follows. The \emph{length} of a substring is the string-depth of its associated node. The \emph{number of occurrences} is the number of leaf descendants.
  Both quantities can be computed simultaneously for every node using a bottom-up traversal.
  For every node corresponding to a frequent substring, we proceed as follows.
  If the node is implicit, we make it explicit. 
  By \Cref{thm:string}, this adds only $\cO(n)$ extra nodes to the tree.
  Then, we query~$D$ using this substring to compute the $w$-size answer and store it with the corresponding node in the suffix tree. This requires $\cO(|D|+w \cdot n)$ space.

  Upon a query $P$, we first locate the corresponding node in our suffix tree in $\cO(m)$ time~\cite{DBLP:conf/focs/Weiner73}. If the node is annotated with a precomputed answer, we output the answer in $\cO(w)$ time. If not, we offload the query to $D$, which outputs the answer in $\cO(m + |\Occ_T(P)| + w)$ time (by the hypothesis).
  As $P$ was not frequent (otherwise, there would have been a precomputed answer), we know that $m \ge |\Occ_T(P)|$ and thus this case also takes $\cO(m + w)$ time.
\end{proof}

Another large category of problems where \Cref{thm:ds-blackbox} can be applied is where the desired output of the query is a summary statistic over some or all of the occurrences of the pattern $P$ in $T$. The following corollary shows that for many of these cases, we can achieve optimal query time, given that a few mild conditions are met. 
Yet, it does not exhaustively characterize these use cases.
In some cases, a more careful, specialized application can yield good results, even if the meta statement below is not applicable (we discuss some such cases in \Cref{sec:other-apps}).

\begin{corollary}[Summary Statistics]
  \label{cor:ds-stat}
  Let $T$ be a text of length $n$
  and $\Stat(T,i,j)$ be an $\cO(1)$-size statistic computable, for any fragment $T[i\dd j]$, in $\cO(1)$ time after preprocessing using space $s_\Stat$.
  Let $\aggOp$ be an associative and commutative aggregation function for the output of $\Stat$ 
  that can combine two values in $\cO(1)$ time.
  There is an $\cO(n+s_\Stat)$-space data structure that,
  given a pattern $P$ of length $m$,
  computes $\aggOp_{i \in \Occ_T(P)} \Stat(T,i, i+\setsize{P}-1)$ in $\cO(m)$ time.
\end{corollary}

We see an application of this corollary in \Cref{sec:u-string} and many more in \Cref{sec:other-apps}. Furthermore, in \Cref{sec:construct}, we show how to construct the data structure in $\cOt(n \sqrt{n})$ time for a class of problems in this framework.

\begin{proof}
  We will construct a data structure $D$ that fulfills the requirements from \Cref{thm:ds-blackbox}. The result then follows from \Cref{thm:ds-blackbox}.
  
  As preprocessing, we construct the suffix tree of $T$ (exactly as we do in \Cref{thm:ds-blackbox}) and, in addition, do the preprocessing for $\Stat$. The data structure has size $|D|=\cO(n + s_\Stat)$.

  Upon a query $P$, we find all occurrences of $P$ in $T$ in $\cO(m+|\Occ_T(P)|)$ time using the suffix tree of~$T$, and,
  for each occurrence $i\in \Occ_T(P)$, we compute $\Stat(T,i, i+m-1)$ in $\cO(1)$ time each.
  We next combine these $\setsize{\Occ_T(P)}$ values in $\cO(\setsize{\Occ_T(P)})$ time using repeated applications of $\aggOp$, obtaining the desired output in total query time $\cO(m + |\Occ_T(P)|)$.
\end{proof}

\begin{remark}
If the querying algorithm of $D$ from \Cref{thm:ds-blackbox} or the functions $\Stat$ or $\aggOp$ from \Cref{cor:ds-stat} have worse running times, the resulting data structure after applying our results degrades gracefully.
In particular, if they require logarithmic time, then we only get an extra logarithmic factor in our query time as well.
This is easy enough to verify by adding this overhead to the calculations in the proofs of \Cref{thm:ds-blackbox,cor:ds-stat}.
We omitted this to simplify the presentation (most importantly, to avoid overly complex statements).
\end{remark}

\section{Applications}\label{sec:headline}

We first present three applications of our framework, where it either yields a novel counting structure for a widely-studied reporting problem, improves upon the best known data structure for a problem, or enriches an existing index with natural non-trivial query types. Then, we recover several known results, thus unifying the treatment of a number of problems.

\subsection{Counting Consecutive Occurrences}
\label{sec:consecutive} 

Given a string $T$, an integer interval $[\alpha,\beta]$, and a pattern $P$, a pair $(i,j)$ is called an $(\alpha,\beta)$-\emph{consecutive occurrence} of $P$ in $T$ if $i,j\in \Occ_T(P)$, there is no other occurrence of $P$ between $i$ and $j$, and $\alpha \leq j - i \leq \beta$. We denote this set of pairs by $\Occco_T(P)$. Many variants of consecutive occurrences have been studied in the literature~\cite{DBLP:journals/algorithmica/BilleG14,
DBLP:conf/stacs/BilleGLPRS24,
DBLP:journals/tcs/BilleGPRS22,
DBLP:journals/algorithmica/BilleGPS23,
DBLP:journals/mst/BilleGVV14,
DBLP:conf/cpm/BrodalLPS99,
DBLP:journals/is/GawrychowskiGSS26,
DBLP:journals/algorithmica/IliopoulosR09,
DBLP:conf/wads/KellerKL07,
DBLP:conf/soda/Muthukrishnan02,
DBLP:journals/tcs/NavarroT16,
DBLP:conf/esa/Radoszewski23}. 
Here, we consider the counting version of the most basic variant.

\defDSproblem{Consecutive Occurrences Counting (\COC)}
{A text $T$ of length $n$ and an integer interval $[\alpha,\beta]$.}
{Given a pattern $P$ of length $m$, return $|\Occco_T(P)|$.}

Applying \Cref{thm:ds-blackbox} to the $\cO(n \log n)$-space data structure of Navarro and Thankachan~\cite{DBLP:journals/tcs/NavarroT16} for optimal reporting of $\Occco_T(P)$ yields an $\cO(n \log n)$-space data structure with $\cO(m)$-time queries. 
However, we can do better. We rely on the following auxiliary data structure.

\begin{lemma}[\cite{DBLP:conf/isaac/BrodalFGL09}]\label{lem:sort-range}
    Given an array $A$ of $n$ elements, there exists an $\cO(n)$-space data structure, so that given integers $i,k>0$, it  reports the elements of $A[i \dd i+k-1]$ in sorted order in the optimal $\mathcal{O}(k)$ time.
\end{lemma}

\begin{theorem}
    \label{the:coc}
    There exists an $\cO(n)$-space data structure for \COC with $\cO(m)$ query time.
\end{theorem}
\begin{proof}
    We construct the suffix tree of $T$, storing the labels of outgoing edges at each node using \Cref{lem:dict}, and its suffix array preprocessed using \Cref{lem:sort-range}.
    We store the answers for the at most $n$ substrings $U$ of~$T$ with $|\Occ_T(U)|>\setsize{U}$ (\Cref{thm:string}) in the suffix tree.
    The total space is $\cO(n)$.
    Upon a query $P$, if the answer is stored, we return it in $\cO(m)$ time by spelling $P$ in the suffix tree, thus reaching the node where it is stored.
    Otherwise, by using \Cref{lem:sort-range}, we obtain the \emph{sorted} array $A$ of the at most $m$ occurrences of $P$ in~$T$ using the corresponding suffix array range in $\cO(|A|)$ time.
    Consecutive occurrences of $P$ in~$T$ correspond to adjacent elements in $A$; we iterate over $A$ in $\cO(|A|)$ time, incrementing our count whenever $\alpha \le A[k+1] - A[k] \le \beta$. Since $|A| \le m$, the query time is $\cO(m)$.
\end{proof}

Interestingly, we can lift the restriction that \emph{there is no in-between occurrence}, and count \emph{all pairs} $(i,j)$, such that $i,j\in \Occ_T(P)$ and $\alpha \leq j - i \leq \beta$. These are known as \emph{$(\alpha,\beta)$-gapped occurrences}~\cite{DBLP:conf/stacs/BilleGLPRS24}.
For a fixed $j$, the value $(A[j]-A[i])$ decreases as $i$ increases, so the valid indices $i<j$ form a contiguous block. That lets us maintain two pointers (``fingers'') across the whole scan instead of restarting a search for every $j$. The query time is thus still $\cO(m)$.    

\subsection{Strings with Utilities}
\label{sec:u-string}

A \emph{string with utilities} (u-string, in short) is a pair $(T,w)$, consisting of a string $T\in\Sigma^n$ and a function $w$ mapping each $i\in[1,n]$ to a real number called the \emph{utility} of $T[i]$.
The \emph{local utility} of a fragment $T[i\dd j]$ of~$T$ is defined as $\sum_{k=i}^j w[k]$. The \emph{global utility} of a pattern $P\in\Sigma^m$ in $T$ is defined as $\sum_{i\in\Occ_T(P)} \sum_{j=i}^{i+m-1} w[j]$.
See~\cite{DBLP:journals/tkde/GanLFCTY21} for a survey.
The following problem of indexing u-strings was introduced and studied by Bernardini et al.~\cite{DBLP:conf/icde/BernardiniCCGGLPP25}. 

\defDSproblem{Useful String Indexing (\USI)}
{A u-string $(T,w)$ of length $n$.}
{Given a pattern $P$ of length $m$, return its global utility in $T$.}

Bernardini et al.~\cite{DBLP:conf/icde/BernardiniCCGGLPP25}
presented an $\cO(n^2/\tau)$-space data structure
with $\cO(m+\tau)$-time queries. 
We present a substantial improvement:

\begin{theorem}
    \label{the:usi}
    There exists an $\cO(n)$-space data structure for USI with $\cO(m)$ query time.
\end{theorem}
\begin{proof}
    We use \Cref{cor:ds-stat} setting $\Stat(T,i, j)\coloneqq \sum_{k=i}^j w[k]$ and $\aggOp \coloneqq \sum$.
    As adding two numbers takes constant time, it suffices to show how to preprocess $T$ and $w$ to enable $\cO(1)$-time access to $\Stat$.
    This is possible using the prefix sums of $w$.
    That is, in preprocessing, we construct $\Pi_w[i] \coloneqq \sum_{j \le i} w[j]$, for all $i \in [1,n]$.
    Then, for an occurrence $P=T[i\dd j]$, we can compute $\sum_{k=i}^j w[k] = \Pi_w[j]-\Pi_w[i-1]$ in $\cO(1)$ time.
\end{proof}

For simplicity of presentation, we have used the \emph{addition} operator for both the local and the global utility in the above definitions. We will generalize \Cref{the:usi} in \Cref{sec:construct}.

\subsection{Weighted Sequences}
\label{sec:weighted-sequences}

    In a \emph{weighted sequence} $T$ of length $n$ over an alphabet $\Sigma$, 
    in every position $i\in[1,n]$, every letter of the alphabet is 
    associated with a probability of occurrence such that the sum of probabilities at each position equals $1$.
    We denote by $p_i^{(T)}(c)$ the occurrence probability of letter $c$ at position $i$ of $T$.
    The classical pattern matching problem extends naturally to weighted sequences.
    Given a string $P$ called a pattern, a weighted sequence~$T$ called a text, 
    both over an alphabet~$\Sigma$, and a threshold probability $\frac1z$, the task is to find all
    positions~$i$ in~$T$ where $\text{Prob}(P,i):=\prod^{\setsize{P}-1}_{j=0}p_{i+j}^{(T)}(P[j]) \geq \frac1z$; $\text{Prob}(P,i)$ is called the occurrence probability of~$P$ at position $i$ of $T$.
    Barton et al.~\cite{DBLP:journals/iandc/BartonK0PR20} showed that we can construct an index of $\cO(nz)$ words of space
    supporting pattern matching queries in the optimal time $\cO(m+|\Occ^{z}_T(P)|)$, where $\Occ^{z}_T(P)$ is the set of occurrences of $P$ in $T$.
    To the best of our knowledge, there is no known index that can find the heaviest occurrence of $P$ in $T$ among the ones in $\Occ^{z}_T(P)$ or the average probability of the occurrences in $\Occ^{z}_T(P)$ in optimal time.   
    The two corresponding problems are formally defined as follows.

    \defDSproblem{Heaviest $z$-Occurrence (\HO)}
{A weighted sequence $T$ of length $n$ and a threshold $\frac1z$.}
{Given a pattern $P$ of length $m$, return $\max\{\text{Prob}(P,i) \mid i \in \Occ^{z}_T(P)\}$.}

\defDSproblem{Average $z$-Occurrence Probability (\AOP)}
{A weighted sequence $T$ of length $n$ and a threshold $\frac1z$.}
{Given a pattern $P$ of length $m$, return $\sum_{i\in \Occ^{z}_T(P)} \text{Prob}(P,i)/|\Occ^{z}_T(P)|$.}
    
    \begin{theorem}
        \label{the:weighted}
        There exists an $\cO(nz)$-space data structure that can answer $\HO$ or $\AOP$ queries in $\cO(m)$ time.
    \end{theorem}
    \begin{proof}
        Barton et al.~\cite{DBLP:journals/iandc/BartonK0PR20} showed the following: for every weighted sequence $T$ of length $n$ and every~$\frac1z$, we can construct a family $\mathcal{F}(T)=(S_j,w_j)_{j=1}^{\lfloor z \rfloor}$ of $\lfloor z \rfloor$ u-strings, each of length~$n$, that carries all the information about the strings occurring in $T$ with probability $p\geq \frac1z$. Formally, $P$ occurs at position $i$ of $T$ with occurrence probability $p\geq \frac1z$ if and only if
        there exists a string $S_j\in \mathcal{F}(T)$ such that $P=S_j[i\dd i+m-1]$, $\prod^{i+m-1}_{k=i}w_j[k]=p$, and $p\geq \frac1z$.
    
        We first construct the suffix tree of $Y=S_1\$\dots \$S_{\lfloor z \rfloor}\$$,
        where $\$$ is a unique letter that does not occur in $S_1,\ldots,S_{\lfloor z \rfloor}$. We also store the labels of outgoing edges at each node using \Cref{lem:dict}.
        This takes $\cO(nz)$ space.
        Using this, upon a query, we can compute $\Occ_Y(P)$ in $\cO(m + |\Occ_Y(P)|)$ time.
        Similarly to the proof of~\Cref{the:usi}, we can also, after preprocessing, access the probability of each of these occurrences in $\cO(1)$ time.
        From this, we directly get an $\cO(nz)$-space data structure to answer $\HO$ or $\AOP$ queries in $\cO(m + |\Occ_Y(P)|)$ time.
        Applying \Cref{thm:ds-blackbox} to this yields the claimed trade-off.
    \end{proof}

\subsection{More Applications}
\label{sec:other-apps}

Beyond the three headline applications described above, the combinatorial insight (\Cref{thm:string}) leads to efficient data structures for a wide range of summary statistics queries.
Many of them have a known efficient data structure, so our main contribution here is that \Cref{thm:ds-blackbox} and \Cref{cor:ds-stat} provide a \emph{simple} and \emph{unified} framework for this set of problems.

For instance, for $\Stat(T,i, j) \coloneqq 1$ and $\aggOp \coloneqq \sum$, we get the classical counting of occurrences of $P$ in~$T$; for $\Stat(T,i, j) \coloneqq \sum^j_{k=i} w[k]$ (for a weight array $w$ of length $n=\setsize{T}$) and $\aggOp \coloneqq \sum$, we solve the \USI problem.
This directly extends to combinations of these statistics \eg{the average weight of the occurrences of the pattern}.

Let us give a few more examples.
Assume that $w$ is a length-$n$ weight array given at preprocessing time and denote by $m$ the length of $P$.

\begin{example}[Heaviest Occurrence]
  Set $\Stat(T,i, j) \coloneqq \sum^j_{k=i} w[k]$ and $\aggOp \coloneqq \max$. \Cref{cor:ds-stat} yields an $\cO(n)$-space data structure answering such queries in $\cO(m)$ time. 
This is a natural follow-up to the \USI problem where the heaviest occurrence might then be a particularly important one for a given application.
Also note that the data structure can be straightforwardly extended to output the position of the heaviest occurrence instead of only the weight.
Other statistics that can be obtained in a similar fashion include the classical problems of finding the first (leftmost) or (rightmost) last occurrence of a given pattern.
\end{example}

\begin{example}[Higher Moments]
  Let $r \ge 0$. Set $\Stat(T,i, j) \coloneqq (\sum^j_{k=i} w[k])^r$ and $\aggOp \coloneqq \sum$. \Cref{cor:ds-stat} yields an $\cO(n)$-space data structure answering such queries in $\cO(m)$ time. 
Observe that for $r=0$, we recover counting of occurrences; for $r=1$, we recover \USI. For $r=2$, we get the sum of squares of occurrence weights.
Since the variance of a random variable $X$ is $\E[X^2] - \E[X]^2$, we can use the data structure for $r=1$ and $r=2$ to deduce the variance of the occurrence weights of a pattern. This is particularly well-motivated within the framework of strings with utilities~\cite{DBLP:conf/sdm/0001CCGGLPP24,DBLP:conf/icde/BernardiniCCGGLPP25}, where $w$ quantifies the profit, risk, or confidence score of each position in $T$. 
Higher moments $r\ge3$ serve as fundamental building blocks for more complex statistics~\cite{DBLP:journals/dam/Regnier00} and we thus note that they are supported by the framework.
\end{example}

\begin{example}[Pattern Matching with Properties]
In this problem, 
in addition to $T$ we are given a \emph{hereditary property}~$\Pi$, which is a family of integer intervals contained in $[1, n]$ (hereditary 
means that it is closed under subintervals). Our goal is to preprocess $T$ 
so that, for a pattern $P$, we can report all occurrences of $P$ in $T$ 
which, interpreted as intervals, belong to $\Pi$. The property $\Pi$ can be 
represented in $\cO(n)$ space using an array $L$ such that the longest interval 
in $\Pi$ starting at position $i$ is $[i, i+L[i]-1]$. The problem was studied in~\cite{DBLP:journals/tcs/AmirCIKZ08,DBLP:journals/iandc/BartonK0PR20,DBLP:journals/jea/Charalampopoulos20,DBLP:journals/ipl/IliopoulosR08,DBLP:journals/ipl/JuanLW09}; it is closely related to indexing weighted sequences. We obtain a general analogue to \Cref{cor:ds-stat} by wrapping any $\Stat$ function in a check to only consider occurrences that obey the property restriction \eg{\cite{DBLP:journals/iandc/BartonK0PR20} shows how to perform this check optimally}.
\end{example}

\begin{example}[Non-Overlapping Occurrences] Given a string~$S$ and two fragments of $S$, $S[i\dd i']$ and $S[j\dd j']$ with $i \leq j$, we say that the fragments are \emph{non-overlapping} if $j > i'$.
The set of non-overlapping occurrences of a string $P$ in a string $T$ is a \emph{largest} set of occurrences of $P$ in $T$ that are pairwise non-overlapping.
We denote this set by $\Occno_T(P)$.
There exists a well-known $\cO(n)$-space data structure with optimal $\cO(m+|\Occno_T(P)|)$ query time for the reporting version~\cite{DBLP:conf/isaac/CohenP09}. Thus, applying \Cref{thm:ds-blackbox} to the reporting data structure directly yields an $\cO(n)$-space data structure with $\cO(m)$-time queries for the counting version. 

The counting version can alternatively be solved via the minimal augmented suffix tree (MAST)~\cite{DBLP:journals/algorithmica/ApostolicoP96,DBLP:conf/icalp/BrodalLOP02}.
Note that the solution of~\cite{DBLP:conf/isaac/CohenP09} is not based on MAST, so together with \Cref{thm:ds-blackbox} this yields an alternative construction for the counting version.
\end{example}

\begin{example}[Covered Positions]\label{ex:covers}
  Given a pattern $P$, we want to count the positions in~$T$ that are in some occurrence of~$P$. We cannot deduce this from the number of occurrences, as we do not want to double-count positions that are in more than one occurrence. It is easy to obtain an $\cO(n)$-space data structure with $\cO(m+|\Occ_T(P)|)$-time queries using the suffix tree and the suffix array preprocessed using \Cref{lem:sort-range}: on a query $P$, we iterate over all \emph{sorted} occurrences and for two consecutive occurrences $i<j$, we add $\min(m,j-i)$ to the count. Applying \Cref{thm:ds-blackbox} yields an $\cO(n)$-space data structure with $\cO(m)$-time queries.  
  
  Note that for this task, an $\cO(n)$-space data structure with optimal query time is known based on the cover suffix tree~\cite{DBLP:journals/talg/KociumakaKRRW20,DBLP:conf/esa/Radoszewski23}. This data structure is more intricate (and has other uses, too), so we think this is a cute simplification for this task.
\end{example}

\begin{example}[Overlapping Occurrences] 
  It is also natural to ask for the maximum number of occurrences that \emph{do} overlap.
  Two ways to make this precise are to ask for the size of a largest set of occurrences such that either: (a) they all mutually overlap; or (b) they all transitively overlap \ie{two occurrences do not need to overlap if they are connected by some sequence of overlaps with other occurrences from the set}. For both of these flavors, we obtain $\cO(n)$-space data structures with optimal query times using \Cref{thm:ds-blackbox}. For $\cO(m+|\Occ_T(P)|)$ query time, simply find the sorted set of occurrences of $P$ as in \Cref{ex:covers}. Then perform a scanline algorithm, incrementing a counter every time an occurrence starts and decrementing it when one ends. The maximum value of this counter yields (a), and the maximum number of occurrences processed between two consecutive points where the counter is zero yields (b), both of which are simple enough to track in $\cO(\setsize{\Occ_T(P)})$ time.
\end{example}

\section{Generalized Combinatorial Lemma and Applications}\label{sec:general}

\begin{lemma}
\label{lem:tree-general}
    For any rooted tree $\mathcal{T}$ with $L$ leaves over a set $\mathcal{V}$ of nodes, and $\tau \in [1,L]$, we have
    \[
    |\{v\in \mathcal{V} \mid \tau \cdot \textsf{depth}(v) < |\mathcal{L}(v)|, \; \textsf{depth}(v) > 0\}|\leq L/\tau,
    \]
    where $\textsf{depth}(v)$ is the depth of node $v$ and $\mathcal{L}(v)$ is the set of leaf descendants of $v$.
\end{lemma}
\begin{proof}
    Let $\mathcal{S} = \{v\in \mathcal{V} \mid \tau \cdot \textsf{depth}(v) < |\mathcal{L}(v)|, \; \textsf{depth}(v) > 0\}$. By definition, we have $|\mathcal{L}(v)| \ge \tau$ for every $v \in \mathcal{S}$. We will define a function $\phi$ that maps each node in $\mathcal{S}$ to a set of exactly $\tau$ of its leaf descendants such that $\phi(u) \cap \phi(v) = \emptyset$ for every distinct $u,v \in \mathcal{S}$. Since there are only $L$ leaves in total, this immediately leads to a bound of $|\mathcal{S}| \le L/\tau$.

    For each node $v$, let $\ell_{v,1},\ell_{v,2},\ldots,\ell_{v,|\mathcal{L}(v)|}$ be the leaves in $\mathcal{L}(v)$ ordered by their rank in an in-order traversal of $\mathcal{T}$. For each node $v \in \mathcal{S}$, we set
    \[
    \phi(v) := \{ \ell_{v, i} \mid i = (\textsf{depth}(v)-1) \cdot \tau + j, \text{ for } j \in [1,\tau]\}.\]  
    
    We consider the following two cases:

    \proofsubparagraph{\boldmath{Case 1: Neither $u$ nor $v$ is an ancestor of the other.}} Because the subtrees rooted in $u$ and~$v$ are disjoint, we have $\mathcal{L}(u) \cap \mathcal{L}(v) = \emptyset$. Since $\phi(u) \subseteq \mathcal{L}(u)$ and $\phi(v) \subseteq \mathcal{L}(v)$, it immediately follows that $\phi(u) \cap \phi(v) = \emptyset$.

    \proofsubparagraph{\boldmath{Case 2: $u$ is an ancestor of $v$ or vice versa.}} Assume without loss of generality that $u$ is an ancestor of~$v$. In this case, the leaves $\ell_{v,1},\ldots,\ell_{v,|\mathcal{L}(v)|}$ appear as some contiguous interval within $\ell_{u,1},\ldots,\ell_{u,|\mathcal{L}(u)|}$. In other words, there exists some $k \ge 0$ such that $\ell_{v,i} = \ell_{u,i+k}$ for all $i \in [1,|\mathcal{L}(v)|]$. Suppose that there exists some leaf $w \in \phi(u) \cap \phi(v)$. By the definition of $\phi$, this means $\ell_{v,(\textsf{depth}(v)-1) \cdot \tau + j_1} = \ell_{u,(\textsf{depth}(u)-1) \cdot \tau + j_2}$ for some $j_1,j_2 \in [1,\tau]$. From this, we get $k = ((\textsf{depth}(u)-1) \cdot \tau + j_2) - ((\textsf{depth}(v)-1) \cdot \tau + j_1) = \tau \cdot (\textsf{depth}(u) - \textsf{depth}(v)) + j_2 - j_1$. Because $\textsf{depth}(v) > \textsf{depth}(u)$ and $j_2 - j_1 < \tau$, this leads to $k < 0$, a contradiction. As a consequence, the leaf $w$ cannot exist.
    This means $\phi(u) \cap \phi(v) = \emptyset$.

    Because $\phi(u) \cap \phi(v) = \emptyset$ for all distinct $u,v \in \mathcal{S}$, we have that $|\bigcup_{v \in \mathcal{S}} \phi(v)| = \sum_{v \in \mathcal{S}} |\phi(v)| = |\mathcal{S}| \cdot \tau$. Moreover, since there are $L$ leaves in total, we have $|\bigcup_{v \in \mathcal{S}} \phi(v)| \le L$ which leads to the bound $|\mathcal{S}| \cdot \tau \le L$ and thus $|\mathcal{S}| \le L / \tau$.
\end{proof}

\stringcor*
\begin{proof}
    In the trie $\{T[1\dd n]\$,T[2\dd n]\$,\ldots,T[n]\$\}$, for a unique letter $\$$ that does not occur in $T$, each internal node~$v$ corresponds one-to-one to a substring $U$ of $T$, with $|U| = \textsf{depth}(v)$. Moreover, the leaves $\mathcal{L}(v)$ correspond to occurrences of $U$.
    As there are $n$ leaves in total, the statement follows directly by applying \Cref{lem:tree-general} to the trie.
\end{proof}

Analogously to the applications of \Cref{thm:string},
\Cref{cor:general} can be applied for designing
space-efficient data structures for counting queries.
We next showcase how we can obtain a \emph{compact} data structure for counting non-overlapping occurrences in near-optimal time, which was mentioned as a ``challenging'' open problem by Gibney et al.~\cite{DBLP:journals/tcs/GibneyMT25}. 

Ganguly et al.~\cite{DBLP:journals/algorithmica/GangulyST20} showed that all we need for reporting non-overlapping occurrences is a suffix tree (or any of its space-efficient counterparts) of text $T$.
The time complexity of their query algorithm is 
$\cO(\mathrm{search}(P)+t_{\textsf{SA}}|\Occno_T(P)|+\mathrm{sort}_n(|\Occno_T(P)|))$, where $\mathrm{search}(P)$
denotes the time for computing the suffix array range of pattern $P$, 
$t_{\textsf{SA}}$ denotes the time for accessing a given entry in the suffix array (or inverse suffix array), and $\mathrm{sort}_n(x)$
denotes the time for sorting a subset of $\{1,2,\ldots,n\}$ of size $\cO(x)$.

\begin{theorem}
    Given a text $T$ of length $n$ over an integer alphabet $[0,\sigma)$, we can construct
    a data structure of $\cO(n \log \sigma)$ bits
    so that, given a pattern $P$ of length $m$,
    it outputs $|\Occno_T(P)|$ in $\cO(m\log n\log_{\sigma}^\epsilon n)$ time for an arbitrarily small constant $\epsilon>0$.
\end{theorem}
\begin{proof}
Let us set $\tau:=\lceil\log_\sigma n \rceil$. 
By \Cref{cor:general}, the number of substrings $U$ of $T$ that satisfy $\setsize{U} \cdot \tau < |\Occ_T(U)|$ is at most $\lfloor n/\tau \rfloor$.
We store the answer for each such substring using the deterministic static dictionary from~\cref{lem:dict}, keyed by the pattern's suffix array range start index and its length.
Since each key-value pair requires $\cO(\log n)$ bits, the dictionary takes $\cO(n \log \sigma)$ bits of space.
Then, for any pattern~$P$ of length $m$, either the answer is stored or
we have at most $m \lceil \log_\sigma n \rceil$ occurrences of $P$ in $T$.
Using the reporting data structure of Ganguly et al.~\cite{DBLP:journals/algorithmica/GangulyST20}, we achieve a query time of $\cO(\mathrm{search}(P)+t_{\textsf{SA}}m\log_\sigma n+\mathrm{sort}_n(m\log_{\sigma} n))$.
Using the compressed suffix array~\cite{DBLP:journals/siamcomp/GrossiV05,DBLP:journals/mst/Sadakane07}, we attain $t_{\textsf{SA}}=\cO(\log^{\epsilon}_{\sigma} n)$ using $\cO(n \log \sigma)$ extra bits of space.
Using binary search as in~\cite{DBLP:journals/siamcomp/ManberM93} over the suffix array~\cite{DBLP:journals/siamcomp/GrossiV05} with the help of the LCP array~\cite{DBLP:journals/mst/Sadakane07}, we get $\mathrm{search}(P)=\cO(m+\log n \log^{\epsilon}_{\sigma} n) \in \cO(m\log n\log_{\sigma}^\epsilon n)$.
Finally, using fast integer sorting~\cite{DBLP:conf/stoc/Han02}, we get $\mathrm{sort}_n(m\log_\sigma n)=\cO(m\log_\sigma n \log\log n)$.
The total space usage is $\cO(n \log \sigma)$ bits and the query time is $\cO(m\log n\log_{\sigma}^\epsilon n)$.
\end{proof}

\section{Another Generalization of the Combinatorial Lemma with Applications in the IPM Model}\label{sec:internal}

\inflatecor*
\begin{proof}
    Consider the string $S$ formed by concatenating $\tau$ copies of $T$ separated by a unique letter $\$$ that does not occur in $T$. In particular, consider the string:
\[
S := (T\$)^{\tau} = \underbrace{T\$T\$\dots\$T\$}_{\tau\text{ copies of } T\$}.
\]
The total length of $S$ is $|S| = \tau(n+1)$.

Let $U$ be a non-empty substring of $T$ that satisfies the condition in the statement, namely $|U| < \tau \cdot |\Occ_T(U)|$. Because the string $U$ does not contain $\$$, it cannot overlap the boundaries between the concatenated copies of $T$ in $S$. Therefore, every occurrence of $U$ in $S$ must be fully contained within exactly one of the $\tau$ individual copies of $T$. This yields:
\[|\Occ_S(U)| = \tau \cdot |\Occ_T(U)|.\]
Substituting this relation into our inequality, we conclude that $U$ must satisfy:
\[|U| < |\Occ_S(U)|.\]

Consider the suffix trie of $S$ (excluding the $\tau$ suffixes of $S$ starting with $\$$). This trie has exactly $n\tau$ leaves. By applying \Cref{lem:tree} to this trie, we conclude that the number of non-empty substrings $U$ satisfying the statement's condition is at most $n\tau$.
\end{proof}

We apply \Cref{cor:inflate}
to \USI in the IPM model. Namely, the pattern $P$ is now specified in $\cO(1)$ time and space using two integers $i$ and $j$ such that $P=T[i\dd j]$ and $m=\setsize{P}=j-i+1$.
We term the corresponding queries \emph{internal} \USI queries, and show the following result.

\begin{theorem}
For every string $T$ of length $n$ and every integer $\tau \in [1,n]$, there exists an $\cO(n\tau)$-space data structure for \USI that can answer internal \USI queries in time $\cO\bigl(1 + m/\tau\bigr)$.
\end{theorem}
\begin{proof}
We construct the suffix tree and the suffix array of $T$. 
We also preprocess the suffix tree
for answering weighted ancestor queries in $\cO(1)$ time using $\cO(n)$ space~\cite{DBLP:conf/cpm/BelazzouguiKPR21,DBLP:conf/esa/GawrychowskiLN14}. A \emph{weighted ancestor query} (WAQ) locates the locus of a substring $T[p\dd q]$ in the suffix tree of~$T$: given the node $u$ of the tree corresponding to $T[p\dd n]$, the query locates the highest ancestor~$v$ of $u$, such that the string written on the root-to-$v$ path has length at least $q - p + 1$.

For every substring $U$ with $|U| < \tau \cdot |\Occ_T(U)|$, we also store the answer using the deterministic static dictionary from~\cref{lem:dict}, keyed by the pattern's suffix array range start index and its length. By \Cref{cor:inflate}, there are $\cO(n\tau)$ such patterns; thus these answers take $\cO(n\tau)$ space. Finally, we compute $\Pi_w[i] \coloneqq \sum_{j \le i} w[j]$, for all $i \in [1,n]$, taking $\cO(n)$ space.

Given a query $(i,j)$, with $m=j-i+1$, we find the locus of $P=T[i\dd i+m-1]$ in the suffix tree of $T$ in $\cO(1)$ time using a WAQ, and distinguish two cases:
\begin{enumerate}
    \item If $|\Occ_T(P)| > m/\tau$, the pattern's exact answer is explicitly stored in the dictionary. We return it in $\cO(1)$ additional time.
    \item If $|\Occ_T(P)| \leq m/\tau$, we traverse the subtree of its locus to locate all $|\Occ_T(P)|$ occurrences. For each occurrence, we compute the required local utility using prefix sums in $\cO(1)$ time. The additional time spent is $\cO(|\Occ_T(P)|) = \cO(m/\tau)$.
\end{enumerate}
In both cases, the extra work is bounded by $\cO(1+m/\tau)$. Thus, the total query time is $\cO(1 + m/\tau)$ using an index of $\cO( n\tau)$ space as claimed.
\end{proof}

\section{Efficient Construction Algorithm}\label{sec:construct}
 
In this section, we show how to efficiently construct the indexing data structure for a class of problems in the framework of \Cref{cor:ds-stat}. 
Let us fix a text $T$ of length $n$. The challenge is to compute statistics for every substring $U$ of~$T$ with $\setsize{U} < |\Occ_T(U)|$. 
We have at most~$n$ such strings~$U$ by \Cref{thm:string}.
Na\"ively querying each of these strings may take $\Theta(n^2)$ time, given that each may have $\Theta(n)$ occurrences.
We show how the associated statistics $\aggOp_{i \in \Occ_T(U)}\Stat(T,i, i+|U|-1)$ can be computed in total $\cOt(n \sqrt{n})$ time using $\cO(n)$ space by bounding the total number of occurrences of short substrings with many occurrences and exploiting the implied periodicity of long ones.

We decompose the problem based on the quantities $\setsize{U}$ and $|\Occ_T(U)|$, by considering:
    \begin{itemize}
    \item short substrings with many occurrences, satisfying $\setsize{U} \le 5\lceil \sqrt{n} \rceil$ and $\setsize{U} < |\Occ_T(U)|$;
    \item long substrings with many occurrences, satisfying $\setsize{U} > 5\lceil \sqrt{n} \rceil$ and $\setsize{U} < |\Occ_T(U)|$.
    \end{itemize}
    
We next describe the class of problems in the framework of \Cref{cor:ds-stat}.

\paragraph{The class of problems.} We follow the framework of \Cref{cor:ds-stat}, with the following additional properties required to make the construction work:
\begin{itemize}
    \item $\Stat(T,i, j)$ can be queried in $\cO(1)$ time after $\cO(n)$-time (and thus $\cO(n)$-space) preprocessing, for all $i$, $j$ (e.g., using prefix sums as in \Cref{the:usi} if the operator is addition, or using range minimum queries~\cite{DBLP:conf/latin/BenderF00} if the operator is the minimum);
    \item $\Stat(T,i, j) = \aggOp_{k=i}^j \Stat(T,k,k)$, for all $i$, $j$. (For simplicity of presentation, we use the same operator as in \Cref{cor:ds-stat}.)
\end{itemize}
Additionally, let us denote by $\otimes$ the repeated application of $\odot$ using a single operand; i.e., $k \otimes x = \underbrace{x \odot \dots \odot x}_k$, for some integer $k>0$. For some operators, such as addition or the minimum, this can be trivially evaluated in $\cO(1)$ time. In general, it can be evaluated in $\cO(\log k)$ time by repeatedly ``squaring'' smaller values. Thus, our bounds generally carry logarithmic factors, which can be shaved if the operator is addition or the minimum.

\paragraph{Preprocessing.} We construct the suffix tree of $T$ in $\cO(n\log n)$ time~\cite{DBLP:conf/focs/Weiner73} using $\cO(n)$ space. Dictionaries storing the labels of outgoing edges to support efficient navigation can be constructed within the same complexities using \Cref{lem:dict}.

The short substrings are handled using \Cref{lem:short}.

\begin{lemma}\label{lem:short}
  Given $\Stat$, the suffix tree of $T$, and an integer $\alpha>0$, we can compute 
  the associated statistic of each substring~$U$ of $T$, with $\setsize{U} \le \alpha$ and $\setsize{U} < |\Occ_T(U)|$, in $\cO(n\alpha)$ total time using $\cO(n)$ space.
\end{lemma}
\begin{proof} We start with a straightforward claim.
\begin{claim}\label{claim:short}  
    The total number of occurrences of substrings $U$ of $T$, with 
    $\setsize{U} \le \alpha$, is $\cO(n\alpha)$.
\end{claim}
\begin{claimproof}
    The number of possible starting positions for a substring of length $m$ in $T$ is $n - m + 1$.
    \[
    \sum_{m=1}^{\alpha} (n - m + 1) = \sum_{m=1}^{\alpha} (n+1) - \sum_{m=1}^{\alpha} m < (n+1)\alpha = \cO(n\alpha). \qedhere
    \]
\end{claimproof}

For every node $v$ in a suffix tree, we denote by
$\str(v)$ the concatenation of the edge labels from the root to $v$.
Using a DFS traversal of the suffix tree of $T$, we store $|\Occ_T(U)|$, for every explicit node $v$, with $U \coloneqq \str(v)$.
This takes $\cO(n)$ time.
Then, for every explicit or implicit node $v$ of the suffix tree,
such that $\setsize{U} \le \alpha$ and $\setsize{U} < |\Occ_T(U)|$,
with $U \coloneqq \str(v)$, we compute the statistic $\aggOp_{i \in \Occ_T(U)} \Stat(T,i, i + |U| - 1)$ as follows. For every occurrence $U=T[i\dd j]$ in $T$, we compute its value in $\cO(1)$ time by the assumption of the framework.
(If $v$ is implicit, we have $\Occ_T(U) = \Occ_T(U')$, with $U' \coloneqq \str(v')$ and $v'$ the closest explicit descendant of $v$.) The construction time is bounded by $\cO(n\alpha)$ by \Cref{claim:short}. By \Cref{thm:string}, we consider at most $n$ nodes, and so we store $\cO(n)$ values.  
\end{proof}

For long substrings, we will exploit the underlying periodic structure.

\begin{definition}[Period]
    The \emph{period} of a string $S=S[1]S[2]\ldots S[|S|]$, denoted by $\per(S)$, is the smallest positive integer $p$ such that $S[i] = S[i + p]$ for all $i \in [1,|S|-p]$.
\end{definition}

A string $S$ is called \emph{periodic} if $\per(S) \le |S| / 2$.
\Cref{lem:long-frequent-periodic} shows that long substrings with many occurrences are in fact highly-periodic.
As a consequence, given a long frequent substring, it suffices to query its $\Stat$  value only within similar periodic fragments of the text.

\begin{lemma}
\label{lem:long-frequent-periodic}
    Let $\beta\in\mathbb{Z}_{>0}$. For any string $U$ with $\setsize{U} \cdot |\Occ_T(U)| > \beta\cdot n$,
    we have $\per(U) \le \frac{1}{\beta} \setsize{U}$.
\end{lemma}
\begin{proof}
Let $k \coloneqq |\Occ_T(U)|$ and $m \coloneqq \setsize{U}$. We have $k \cdot m > \beta\cdot n$. By the hypothesis ($k > \beta\cdot n/m$) and the pigeonhole principle, there exists an index $i$ such that $|\Occ_T(U)\cap [i,i+m-1]|\geq \beta+1$. Thus, there exist two occurrences $j_1<j_2$ of $U$ in $T$, with $j_1,j_2\in[i,i+m-1]$ and $d = j_2-j_1 \leq \frac{m}{\beta}$.
These two occurrences of $U$ at distance $d$ have an overlap of length $m-d$. Thus, $\per(U) \leq d \leq \frac{m}{\beta} = \frac{1}{\beta} \setsize{U}$.
\end{proof}

We make use of the following standard definitions.

\begin{definition}[Lyndon root]
    Given a periodic string $R$ with period $\per(R)$, the \emph{Lyndon root} of $R$ is the lexicographically smallest rotation of $R[1\dd \per(R)]$.
\end{definition}

\begin{definition}[Lyndon representation]
    Given a periodic string $R$ with Lyndon root $L$, the \emph{Lyndon representation} of $R$ is the concatenation $L_1 L^k L_2$, where $L_1$ is a proper suffix of $L$, $L_2$ is a proper prefix of $L$, and $k$ is an integer referred to as the \emph{exponent}.
    Further, we refer to $L^k$ as the \emph{body} of $R$, and $L_1$ and~$L_2$ as the front and back \emph{extensions} of $R$, respectively.
\end{definition}

\begin{definition}[Run]
    A fragment $R=S[i\dd j]$ of a string $S$ is called a \emph{run} of $S$ if $\per(R) \le |R|/2$ (i.e., $R$ is a periodic string) and extending $R$ either to the left or to the right (if possible) would result in an increase of its period, that is, $S[i-1] \ne S[i-1+\per(R)]$ (or $i = 1$) and $S[j+1] \ne S[j+1-\per(R)]$ (or $j=|S|$). 
\end{definition}

We denote by $\mathcal{R}_L(T)$ the set of all runs in a string $T$ that have Lyndon root $L$ and length at least $5\lceil\sqrt{n}\rceil$. By \Cref{lem:long-frequent-periodic}, every occurrence of a long frequent substring $U$ must be a fragment of a run from $\mathcal{R}_L(T)$, where $L$ is the Lyndon root of $U$.

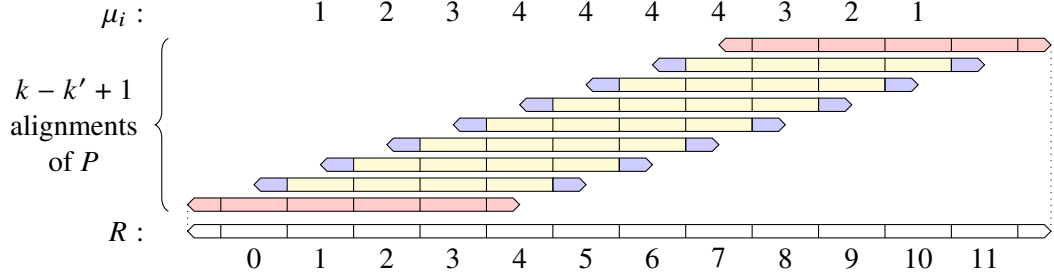
\begin{figure}[htb]
    \centering
    \begin{tikzpicture}[x=25pt, y=-10pt]
        \def\k{12};
        \def\kp{4};
        \def\inners{7}; % \k - \kp - 1
        \def\mumax{4};
        \def\kmmumax{7}; % \k - \mumax - 1

        \def\cOuter{red!20};
        \def\cInner{yellow!20};
        \def\cInnerExt{blue!20};

        \def\yMui{-\inners - 1};
        \def\yI{0.5};
        
        \DrawRun{0}{0}{\k}{white};
        \foreach \y in {1,...,\inners} {
            \DrawRunExt{\y}{-1 - 0.75 * \y}{\kp}{\cInner}{\cInnerExt};
        }
        \DrawRun{0}{-1}{\kp}{\cOuter};
        \DrawRun{\inners + 1}{-1.75 - 0.75 * \inners}{\kp}{\cOuter};

        \foreach \i in {1,...,\k} {
            \ifnum \i>\mumax
                \ifnum \i<\kmmumax
                    \node at (\i + 0.5, \yMui) {$\mumax$};
                \fi
            \fi
        }
        \foreach \i in {1,...,\mumax} {
            \node at (\i + 0.5, \yMui) {$\i$};
            \node at (\k - \i - 0.5, \yMui) {$\i$};
        }

        \foreach \i in {0,...,\k} {
            \ifnum \i<\k
                \node [anchor=north] at (\i + 0.5, \yI) {$\i$};
            \fi
        }
        
        \draw [decorate,decoration={brace,amplitude=5pt,mirror,raise=4ex}]  (0, -0.75 * \inners - 1.75) -- (0, -0.5) node[midway, anchor=east, xshift=-2.5em, align=center]{$k - k' + 1$\\alignments\\of $P$};

        \node[anchor=east] at (-1, 0.25) {$R:$};
        \node[anchor=east] at (-1, \yMui + 0.25) {$\mu_i:$};

        \draw[dotted] (-0.5, -0.75) -- (-0.5, 0.25);
        \draw[dotted] (\k + 0.5, -0.75 * \inners - 1.25) -- (\k + 0.5, 0.25);
    \end{tikzpicture}
    \caption{Occurrences of a pattern $P$ with a run $R$, where $k = 12$ is the exponent of $R$ and $k' = 4$ is the exponent of~$P$. The first occurrence and the last occurrence, both colored red, exist conditionally depending on the extensions of the two runs; the inner occurrences always exist. The bodies of the inner occurrences (yellow) are computed separately from the extensions (blue). In this example, we have $\mu_{\star} = 4$.}
\label{fig:long-preprocess}
\end{figure}

\begin{lemma}\label{lem:runs}
    We are given a periodic fragment $R$ of $T$ by its Lyndon representation. 
    Let $L$ and $k$ denote the Lyndon root and exponent of $R$, respectively.
    We can preprocess $R$ in $\cO(|R| \log k)$ time using $\cO(|R|)$ space so that, given any periodic pattern $P$ with Lyndon root $L$ by its Lyndon representation, we can return $P$'s statistic within $R$ in $\cO(\log k)$ time.
\end{lemma}
\begin{proof}
    Let $R = L_1 L^k L_2$ and $P = L'_1 L^{k'} L'_2$ be the Lyndon representations of $R$ and $P$.
    We hence assume that $k \ge k'$, since otherwise $P$ does not have any occurrences within $R$.

    Every occurrence of $P$ in $R$ starts at a position aligning the occurrences of~$L$ in $P$ with occurrences of $L$ in $R$.
    Note that $L$ is a primitive string (that is, if $L=U^k$ for a string $U$ and an integer $k$, then $U=L$ and $k=1$) and hence it has exactly two occurrences in string $LL$~\cite{DBLP:books/daglib/0020103}.
    An illustration is provided in \Cref{fig:long-preprocess}.

    It then follows that $P$ occurs in~$R$ at least $k - k' - 1$ times and at most $k - k' + 1$ times.
    Note that we can align the first occurrence of $L$ in $P$ with the first (resp.~the $(k-k'+1)$th) occurrence of $L$ in $R$ to get a valid occurrence only if $|L'_1| \le |L_1|$ (resp.~$|L'_2| \le |L_2|$).
    The other occurrences, obtained by aligning the first occurrence of $L$ in~$P$ with the $j$th occurrence of $L$ in $R$ for $j \in [2,k-k']$, which we call \emph{inner occurrences}, exist unconditionally. We compute the statistic of $P$ in $R$ in three steps:
    \begin{enumerate}
        \item the statistic of the bodies (i.e., $L^{k'}$) of $P$'s inner occurrences;
        \item the statistic of the extensions (i.e., $L'_1$ and $L'_2$) of $P$'s inner occurrences;
        \item the statistic of the at most two occurrences of $P$ that are not inner, if they exist.
    \end{enumerate}

    \proofsubparagraph{Inner occurrences body.}
    We assume that $k \geq k'+2$, as otherwise there are no inner occurrences.
    We denote by $X_0,\ldots,X_{k-1}$ the $\Stat$ values of the $k$ occurrences of $L$ in $R$.
    Moreover, we denote by $U_1,\ldots,U_{k-k'-1}$ the $\Stat$ values of the bodies of $P$'s inner occurrences, where $U_i = X_i \odot \dots \odot X_{i+k' - 1}$.
    Combining the values of the inner occurrences, we get:
    \[U_1 \odot \cdots \odot U_{k-k'-1} = (X_1 \odot \cdots \odot X_{k'}) \odot \cdots \odot (X_{k-k'-1} \odot \cdots \odot X_{k-2}) = \aggOp_{i=1}^{k-k'-1}\aggOp_{j=0}^{k'-1}X_{i+j}.\]

    We denote by $\mu_i$ the multiplicity of $X_i$ in the above aggregation.
    As illustrated in \Cref{fig:long-preprocess}, $\mu_i$ depends on $k$, $k'$, and how close the $i$th occurrence of $L$ in $R$ is to an endpoint of the run.
    In particular, we have $\mu_1 = \mu_{k-2} = 1$.
    Then, as we move toward the middle of $R$ (from either direction), the multiplicity might increase, by at most one per occurrence of~$L$, but as each occurrence of $L$ in $R$ belongs to at most $k'$ inner occurrences of $P$ and we have a total of $k-k'-1$ inner occurrences, we have $\mu_i = \min \{k-k'-1, k', i , k-i-1 \}$.
    Let $\mu_\star := \mu_{\lfloor{(k-1)/2}\rfloor}$,
    $y := \arg\min_j \{ j \in [1,k-2] \mid \mu_j = \mu_\star \}$
    and
    $z := \arg\max_j \{ j \in [1,k-2] \mid \mu_j = \mu_\star \}$.
    Note that $y = \min \{k-k'-1, k'\}$, as for $i \geq y$, $\mu_i = \mu_\star$, while for $i<y$, $\mu_i =i$. By symmetry, $z = k - 1 - \min \{k-k'-1, k'\}$. Using the above formula for the multiplicities, we obtain:
    \[
        \aggOp_{i=1}^{k-2} (\mu_i \otimes X_i) = \aggOp_{i=1}^{y-1} (i \otimes X_i) + (\mu_{\star} \otimes \aggOp_{i=y}^{z} X_i) + \aggOp_{i=z+1}^{k-2}((k-i-1) \otimes X_i). 
    \]

    The values of the first and the last aggregations can be precomputed for all possible values of $y$ and~$z$, respectively, in $\cO(k \log k)$ time and stored in $\cO(k)$ space.
    For the middle aggregation, we can also precompute and store all $\cO(k)$ relevant values since $z = k - 1 - y$. 
    We can then evaluate $\mu_{\star} \otimes\aggOp_{i=y}^{z} X_i$ at query time in $\cO(\log k)$ time obtaining the desired value.
    Thus, the values can be computed in $\cO(\log k)$ time after $\cO(|R| \log k)$-time preprocessing using $\cO(|R|)$ space.
    
    \proofsubparagraph{Extensions of inner occurrences.}
    We describe how to compute the values of the front extensions (i.e., $L'_1$) of inner occurrences; the back extensions ($L'_2$) can be handled symmetrically. There are exactly $k - k' - 1$ inner occurrences, whose front extensions are the length-$|L'_1|$ suffixes of the first $k - k' - 1$ occurrences of $L$ in $R$. We denote by $X_i^F$ the value of the length-$|L'_1|$ suffix of the $i$th occurrence of $L$. We have
    \[\aggOp_{i=0}^{k-k'-2}X_i^F = \aggOp_{i=1}^{k-k'-1} \Stat(R,s + pi -|L'_1|, s + pi - 1),
    \]
    where $p = \per(R) = |L|$ and $s$ is the starting position of the first occurrence of $L$ in $R$ (meaning that $s + pi$ is the starting position of the $i$th occurrence of $L$).
    We precompute the values for all relevant $k'$ and $|L'_1|$
    using a recurrence on $k'$.
    There are $k \cdot p$ values to compute in total, so this pre-processing step takes $\cO(kp) = \cO(|R|)$ time.

    \proofsubparagraph{Non-inner occurrences.} The at most two non-inner occurrences of $P$ can either align $P$'s first occurrence of $L$ with $R$'s first occurrence of $L$ or align $P$'s last occurrence of $L$ with $R$'s last occurrence of $L$.
    This entails aligning the front (resp.~back) extensions of the two runs, meaning that they only exist conditionally depending on the lengths of the extensions and the exponents $k$ and $k'$:
    \begin{itemize}
        \item if $k = k'$, then we only have a single non-inner occurrence, which exists if and only if $|L'_1| \le |L_1|$ and $|L'_2| \le |L_2|$;
        \item if $k > k'$, a non-inner occurrence aligning the front extensions exists if and only if $|L'_1| \le |L_1|$ and a non-inner occurrence aligning the back extensions exists if and only if $|L'_2| \le |L_2|$.
    \end{itemize}

    These conditions can be tested in constant time. If the occurrences exist, their values can be retrieved in constant time as well from $\Stat$.
    
    \proofsubparagraph{Wrap-up.} After performing the described precomputations in $\cO(\setsize{R} \log k)$ time using $\cO(\setsize{R})$ space, the statistic of any periodic pattern $P$ can be queried in $\cO(\log k)$ time.
\end{proof}

We make use of the following variant of runs.

\begin{definition}[$\tau$-Run]
A \emph{$\tau$-run} is a run of length at least $3\tau - 1$ with period at most $\frac{1}{3}\tau$.
\end{definition}
\begin{lemma}[\cite{DBLP:conf/esa/Charalampopoulos21}]\label{lem:tau-runs}
For any $\tau\in \mathbb{Z}_{>0}$,
the number of $\tau$-runs in $T\in\Sigma^n$ 
is $\cO(n/\tau)$.    
\end{lemma}

The following result can be used to bound the total length of all highly-periodic runs.

\begin{lemma}[\cite{DBLP:conf/esa/Charalampopoulos20a}]\label{lem:tau-runs-length}
For any $T\in\Sigma^n$, 
$\rho(T)\coloneqq\sum_{R\in\mathcal{R}(T)} (|R|-2\cdot\per(R)+1)=\cO(n\log n)$.
\end{lemma}

We are now ready to describe the algorithm for long substrings.

\begin{lemma}\label{lem:periodic}
   Given $\Stat$, the suffix tree of $T$, and an integer $\alpha\ge 5\lceil \sqrt{n} \rceil$,
  we can compute the statistic of each substring~$U$ of~$T$ 
  with $\setsize{U} > \alpha$ and $\setsize{U} < |\Occ_T(U)|$,
  in $\cOt(n^2 / \alpha)$
  total time using $\cO(n)$ space.
\end{lemma}
\begin{proof}
    For any $U$ with $|\Occ_T(U)| > \setsize{U} > \alpha$, 
    we have $\setsize{U} \cdot |\Occ_T(U)| > \alpha^2$.
    By \Cref{lem:long-frequent-periodic} and the fact that $\alpha \ge 5\lceil \sqrt{n} \rceil$, we know that $\per(U) \le \frac{1}{25}\setsize{U}$.
    Let us consider the set $\mathcal{R}$ of 
    runs~$R$ with $\per(R) \leq \frac{1}{25}\setsize{R}$ and $|R|\geq \alpha$.
    For each run $R \in \mathcal{R}$, there exists some non-negative integer $i$ such that $|R| \geq 3\cdot 2^i \lfloor \alpha /3 \rfloor$ and $\per(R)\leq 2^i \lfloor \alpha /3 \rfloor/3$.
    By \Cref{lem:tau-runs}, the total number of such runs is $\cO(n/\alpha \cdot \sum_{i=0}^\infty 1/2^i) = \cO(n/\alpha)$.
    Furthermore, by \Cref{lem:tau-runs-length}, 
    the total length of all runs $R \in \mathcal{R}$ is in $\cO(n \log n)$.
    We henceforth consider only runs in $\mathcal{R}$.
    We preprocess the suffix tree of $T$ in $\cO(n)$ time for answering WAQs in $\cO(1)$ time~\cite{DBLP:conf/cpm/BelazzouguiKPR21}.

    We compute all runs in $\mathcal{R}$ in $\cO(n)$ time~\cite{DBLP:conf/icalp/Ellert021}.
    This algorithm outputs a canonical representation
    of every run that includes a constant-size representation of its Lyndon root~$L$.
    Using a DFS traversal of the suffix tree of $T$, we compute 
    $|\Occ_T(U)|$, for every explicit node~$v$
    of the suffix tree with $U \coloneqq \str(v)$.
    This takes $\cO(n)$ time.
    Then, during another traversal, we find each explicit or implicit node~$v$ of the suffix tree with path-label $U$ such that $\setsize{U} > \alpha$ and $\setsize{U} < |\Occ_T(U)|$, and insert $U$ to an initially empty set $\mathcal{U}$; we have at most $n$ such nodes by \Cref{thm:string}.
    Note that we can compute the Lyndon representation of a periodic substring~$U$ of $T$ in $\cO(1)$ time after an $\cO(n)$-time preprocessing of $T$ using a so-called 2-period query~\cite{DBLP:journals/siamcomp/KociumakaRRW24} that returns the period of $U$, followed by a minimum rotation query~\cite{DBLP:conf/cpm/Kociumaka16} that returns the lexicographically smallest rotation of $U[1\dd \per(U)]$.

    We group all elements of the multiset $\mathcal{R} \cup \mathcal{U}$ in $\cO(n)$ time using WAQs to compute the loci of their Lyndon roots in the suffix tree (we process the Lyndon roots in the order of increasing length after an $\cO(n)$-time sorting); then two elements are in the same group if their Lyndon roots have the same locus (and the same length).

    We preprocess each $R\in \mathcal{R}$ according to \Cref{lem:runs} and then query the obtained data structure with every $U \in \mathcal{U}$ that has the same Lyndon root as $R$.
    Over all invocations of \Cref{lem:runs}, we have: (i) $\cOt(n)$ total construction time due to the total length of the runs; (ii) $\cO(n)$ space if we process every run separately; and (iii) $\cOt(n^2 / \alpha)$ total query time since $|\mathcal{U}|\leq n$ and $|\mathcal{R}| = \cO(n/\alpha)$.
\end{proof}

\begin{theorem}\label{the:con}
     There exists an $\cO(n)$-space data structure for problems in the framework of \Cref{cor:ds-stat} with $\cO(m)$ query time, given that:
     \begin{itemize}
         \item $\Stat(T,i, j)$ can be queried in $\cO(1)$ time, for all $i$, $j$, after $\cO(n)$-time preprocessing; 
         \item $\Stat(T,i, j) = \aggOp_{k=i}^j \Stat(T,k,k)$, for all $i$, $j$.
     \end{itemize}     
     We can construct this data structure in $\cOt(n\sqrt{n})$ time using $\cO(n)$ space.
\end{theorem}
\begin{proof}  
    For the first part, we simply follow the proof of \Cref{the:usi}, plugging in the more general operators.
    For the second part, preprocessing (i.e., constructing  the suffix tree of $T$) takes $\cO(n\log n)$ time using $\cO(n)$ space.
    By setting $\alpha\coloneqq5\lceil \sqrt{n} \rceil$,
    \Cref{lem:short} takes $\cO(n\sqrt{n})$ time and $\cO(n)$ space, and
    \Cref{lem:periodic} takes $\cOt(n\sqrt{n})$ time and $\cO(n)$ space.
\end{proof}

One can now easily verify that by applying \Cref{the:con} we can construct the data structure underlying \Cref{the:usi} within the claimed bounds.

\bibliographystyle{alphaurl}% the mandatory bibstyle
\bibliography{references}

@inproceedings{DBLP:conf/icde/BernardiniCCGGLPP25,
  author = {Giulia Bernardini and Huiping Chen and Alessio Conte and Roberto Grossi and Veronica Guerrini and Grigorios Loukides and Nadia Pisanti and Solon P. Pissis},
  title = {Indexing Strings with Utilities},
  booktitle = {41st {IEEE} International Conference on Data Engineering, {ICDE} 2025},
  pages = {2782--2795},
  publisher = {{IEEE}},
  year = {2025},
  doi = {10.1109/ICDE65448.2025.00209},
  timestamp = {Wed, 29 Oct 2025 00:00:00 +0100},
  biburl = {https://dblp.org/rec/conf/icde/BernardiniCCGGLPP25.bib},
  bibsource = {dblp computer science bibliography, https://dblp.org},
  _bib2doi_selected = {dblp:/rec/conf/icde/BernardiniCCGGLPP25.bib},
  _bib2doi_confirmed = {true},
}

@inproceedings{DBLP:conf/latin/BenderF00,
  author = {Michael A. Bender and Martin Farach{-}Colton},
  title = {The {LCA} Problem Revisited},
  booktitle = {{LATIN} 2000: Theoretical Informatics, 4th Latin American Symposium},
  series = {Lecture Notes in Computer Science},
  volume = {1776},
  pages = {88--94},
  publisher = {Springer},
  year = {2000},
  url = {https://doi.org/10.1007/10719839\_9},
  doi = {10.1007/10719839\_9},
  timestamp = {Fri, 09 Apr 2021 01:00:00 +0200},
  biburl = {https://dblp.org/rec/conf/latin/BenderF00.bib},
  bibsource = {dblp computer science bibliography, https://dblp.org},
  _bib2doi_selected = {dblp:/rec/conf/latin/BenderF00.bib},
  _bib2doi_confirmed = {true},
}

@inproceedings{DBLP:conf/stoc/Han02,
  author = {Yijie Han},
  title = {Deterministic sorting in ${O}(n\log \log n)$ time and linear space},
  booktitle = {Proceedings on 34th Annual {ACM} Symposium on Theory of Computing,},
  pages = {602--608},
  publisher = {{ACM}},
  year = {2002},
  url = {https://doi.org/10.1145/509907.509993},
  doi = {10.1145/509907.509993},
  timestamp = {Tue, 06 Nov 2018 00:00:00 +0100},
  biburl = {https://dblp.org/rec/conf/stoc/Han02.bib},
  bibsource = {dblp computer science bibliography, https://dblp.org},
  _bib2doi_selected = {dblp:/rec/conf/stoc/Han02.bib},
  _bib2doi_confirmed = {true},
}

@article{DBLP:journals/mst/Sadakane07,
  author = {Kunihiko Sadakane},
  title = {Compressed Suffix Trees with Full Functionality},
  journal = {Theory Comput. Syst.},
  volume = {41},
  number = {4},
  pages = {589--607},
  year = {2007},
  doi = {10.1007/S00224-006-1198-X},
  timestamp = {Wed, 14 Nov 2018 00:00:00 +0100},
  biburl = {https://dblp.org/rec/journals/mst/Sadakane07.bib},
  bibsource = {dblp computer science bibliography, https://dblp.org},
  _bib2doi_selected = {dblp:/rec/journals/mst/Sadakane07.bib},
  _bib2doi_confirmed = {true},
}

@inproceedings{DBLP:conf/icde/Gabory0LPZ24,
  author = {Est{\'{e}}ban Gabory and Chang Liu and Grigorios Loukides and Solon P. Pissis and Wiktor Zuba},
  title = {Space-Efficient Indexes for Uncertain Strings},
  booktitle = {40th {IEEE} International Conference on Data Engineering, {ICDE} 2024},
  pages = {4828--4842},
  publisher = {{IEEE}},
  year = {2024},
  doi = {10.1109/ICDE60146.2024.00367},
  timestamp = {Sun, 06 Oct 2024 01:00:00 +0200},
  biburl = {https://dblp.org/rec/conf/icde/Gabory0LPZ24.bib},
  bibsource = {dblp computer science bibliography, https://dblp.org},
  _bib2doi_selected = {dblp:/rec/conf/icde/Gabory0LPZ24.bib},
  _bib2doi_confirmed = {true},
}

@article{DBLP:journals/jea/Charalampopoulos20,
  author = {Panagiotis Charalampopoulos and Costas S. Iliopoulos and Chang Liu and Solon P. Pissis},
  title = {Property Suffix Array with Applications in Indexing Weighted Sequences},
  journal = {{ACM} J. Exp. Algorithmics},
  volume = {25},
  pages = {1--16},
  year = {2020},
  url = {https://doi.org/10.1145/3385898},
  doi = {10.1145/3385898},
  timestamp = {Sat, 08 Jan 2022 00:00:00 +0100},
  biburl = {https://dblp.org/rec/journals/jea/Charalampopoulos20.bib},
  bibsource = {dblp computer science bibliography, https://dblp.org},
  _bib2doi_selected = {dblp:/rec/journals/jea/Charalampopoulos20.bib},
  _bib2doi_confirmed = {true},
}

@inproceedings{DBLP:conf/esa/GawrychowskiLN14,
  author = {Pawel Gawrychowski and Moshe Lewenstein and Patrick K. Nicholson},
  title = {Weighted Ancestors in Suffix Trees},
  booktitle = {Algorithms - {ESA} 2014 - 22th Annual European Symposium},
  series = {Lecture Notes in Computer Science},
  volume = {8737},
  pages = {455--466},
  publisher = {Springer},
  year = {2014},
  doi = {10.1007/978-3-662-44777-2\_38},
  timestamp = {Fri, 26 May 2017 01:00:00 +0200},
  biburl = {https://dblp.org/rec/conf/esa/GawrychowskiLN14.bib},
  bibsource = {dblp computer science bibliography, https://dblp.org},
  _bib2doi_selected = {dblp:/rec/conf/esa/GawrychowskiLN14.bib},
  _bib2doi_confirmed = {true},
}

@article{DBLP:journals/tcs/GibneyMT25,
  author = {Daniel Gibney and Paul Macnichol and Sharma V. Thankachan},
  title = {Non-overlapping indexing in {BWT}-runs bounded space},
  journal = {Theor. Comput. Sci.},
  volume = {1056},
  pages = {115512},
  year = {2025},
  doi = {10.1016/J.TCS.2025.115512},
  timestamp = {Tue, 14 Oct 2025 01:00:00 +0200},
  biburl = {https://dblp.org/rec/journals/tcs/GibneyMT25.bib},
  bibsource = {dblp computer science bibliography, https://dblp.org},
  _bib2doi_selected = {dblp:/rec/journals/tcs/GibneyMT25.bib},
  _bib2doi_confirmed = {true},
}

@article{DBLP:journals/algorithmica/GangulyST20,
  author = {Arnab Ganguly and Rahul Shah and Sharma V. Thankachan},
  title = {Succinct Non-overlapping Indexing},
  journal = {Algorithmica},
  volume = {82},
  number = {1},
  pages = {107--117},
  year = {2020},
  doi = {10.1007/S00453-019-00605-5},
  timestamp = {Wed, 28 Feb 2024 00:00:00 +0100},
  biburl = {https://dblp.org/rec/journals/algorithmica/GangulyST20.bib},
  bibsource = {dblp computer science bibliography, https://dblp.org},
  _bib2doi_selected = {dblp:/rec/journals/algorithmica/GangulyST20.bib},
  _bib2doi_confirmed = {true},
}

@inproceedings{DBLP:conf/isaac/CohenP09,
  author = {Hagai Cohen and Ely Porat},
  title = {Range Non-overlapping Indexing},
  booktitle = {Algorithms and Computation, 20th International Symposium, {ISAAC} 2009},
  series = {Lecture Notes in Computer Science},
  volume = {5878},
  pages = {1044--1053},
  publisher = {Springer},
  year = {2009},
  doi = {10.1007/978-3-642-10631-6\_105},
  timestamp = {Fri, 19 May 2017 01:00:00 +0200},
  biburl = {https://dblp.org/rec/conf/isaac/CohenP09.bib},
  bibsource = {dblp computer science bibliography, https://dblp.org},
  _bib2doi_selected = {dblp:/rec/conf/isaac/CohenP09.bib},
  _bib2doi_confirmed = {true},
}

@article{DBLP:journals/iandc/BartonK0PR20,
  author = {Carl Barton and Tomasz Kociumaka and Chang Liu and Solon P. Pissis and Jakub Radoszewski},
  title = {Indexing weighted sequences: Neat and efficient},
  journal = {Inf. Comput.},
  volume = {270},
  year = {2020},
  doi = {10.1016/J.IC.2019.104462},
  timestamp = {Sun, 06 Oct 2024 01:00:00 +0200},
  biburl = {https://dblp.org/rec/journals/iandc/BartonK0PR20.bib},
  bibsource = {dblp computer science bibliography, https://dblp.org},
  _bib2doi_selected = {dblp:/rec/journals/iandc/BartonK0PR20.bib},
  _bib2doi_confirmed = {true},
}

@inproceedings{DBLP:conf/cpm/BelazzouguiKPR21,
  author = {Djamal Belazzougui and Dmitry Kosolobov and Simon J. Puglisi and Rajeev Raman},
  title = {Weighted Ancestors in Suffix Trees Revisited},
  booktitle = {32nd Annual Symposium on Combinatorial Pattern Matching, {CPM} 2021},
  series = {LIPIcs},
  volume = {191},
  pages = {8:1--8:15},
  publisher = {Schloss Dagstuhl - Leibniz-Zentrum f{\"{u}}r Informatik},
  year = {2021},
  doi = {10.4230/LIPICS.CPM.2021.8},
  timestamp = {Mon, 03 Mar 2025 00:00:00 +0100},
  biburl = {https://dblp.org/rec/conf/cpm/BelazzouguiKPR21.bib},
  bibsource = {dblp computer science bibliography, https://dblp.org},
  _bib2doi_selected = {dblp:/rec/conf/cpm/BelazzouguiKPR21.bib},
  _bib2doi_confirmed = {true},
}

@inproceedings{DBLP:conf/icalp/Ellert021,
  author = {Jonas Ellert and Johannes Fischer},
  title = {Linear Time Runs Over General Ordered Alphabets},
  booktitle = {48th International Colloquium on Automata, Languages, and Programming, {ICALP} 2021},
  series = {LIPIcs},
  volume = {198},
  pages = {63:1--63:16},
  publisher = {Schloss Dagstuhl - Leibniz-Zentrum f{\"{u}}r Informatik},
  year = {2021},
  doi = {10.4230/LIPICS.ICALP.2021.63},
  timestamp = {Tue, 06 Jul 2021 01:00:00 +0200},
  biburl = {https://dblp.org/rec/conf/icalp/Ellert021.bib},
  bibsource = {dblp computer science bibliography, https://dblp.org},
  _bib2doi_selected = {dblp:/rec/conf/icalp/Ellert021.bib},
  _bib2doi_confirmed = {true},
}

@inproceedings{DBLP:conf/esa/Charalampopoulos20a,
  author = {Panagiotis Charalampopoulos and Jakub Radoszewski and Wojciech Rytter and Tomasz Walen and Wiktor Zuba},
  title = {The Number of Repetitions in {2D}-Strings},
  booktitle = {28th Annual European Symposium on Algorithms, {ESA} 2020},
  series = {LIPIcs},
  volume = {173},
  pages = {32:1--32:18},
  publisher = {Schloss Dagstuhl - Leibniz-Zentrum f{\"{u}}r Informatik},
  year = {2020},
  doi = {10.4230/LIPICS.ESA.2020.32},
  timestamp = {Fri, 21 Nov 2025 23:44:11 +0100},
  biburl = {https://dblp.org/rec/conf/esa/Charalampopoulos20a.bib},
  bibsource = {dblp computer science bibliography, https://dblp.org},
  _bib2doi_selected = {dblp:/rec/conf/esa/Charalampopoulos20a.bib},
  _bib2doi_confirmed = {true},
}

@inproceedings{DBLP:conf/esa/Charalampopoulos21,
  author = {Panagiotis Charalampopoulos and Tomasz Kociumaka and Solon P. Pissis and Jakub Radoszewski},
  title = {Faster Algorithms for Longest Common Substring},
  booktitle = {29th Annual European Symposium on Algorithms, {ESA} 2021},
  series = {LIPIcs},
  volume = {204},
  pages = {30:1--30:17},
  publisher = {Schloss Dagstuhl - Leibniz-Zentrum f{\"{u}}r Informatik},
  year = {2021},
  doi = {10.4230/LIPICS.ESA.2021.30},
  timestamp = {Mon, 26 Jun 2023 01:00:00 +0200},
  biburl = {https://dblp.org/rec/conf/esa/Charalampopoulos21.bib},
  bibsource = {dblp computer science bibliography, https://dblp.org},
  _bib2doi_selected = {dblp:/rec/conf/esa/Charalampopoulos21.bib},
  _bib2doi_confirmed = {true},
}

@inproceedings{DBLP:conf/cpm/Kociumaka16,
  author = {Tomasz Kociumaka},
  title = {Minimal Suffix and Rotation of a Substring in Optimal Time},
  booktitle = {27th Annual Symposium on Combinatorial Pattern Matching, {CPM} 2016},
  series = {LIPIcs},
  pages = {28:1--28:12},
  publisher = {Schloss Dagstuhl - Leibniz-Zentrum f{\"{u}}r Informatik},
  year = {2016},
  doi = {10.4230/LIPICS.CPM.2016.28},
  timestamp = {Thu, 23 Aug 2018 01:00:00 +0200},
  biburl = {https://dblp.org/rec/conf/cpm/Kociumaka16.bib},
  bibsource = {dblp computer science bibliography, https://dblp.org},
  _bib2doi_selected = {dblp:/rec/conf/cpm/Kociumaka16.bib},
  _bib2doi_confirmed = {true},
}

@article{DBLP:journals/siamcomp/KociumakaRRW24,
  author = {Tomasz Kociumaka and Jakub Radoszewski and Wojciech Rytter and Tomasz Walen},
  title = {Internal Pattern Matching Queries in a Text and Applications},
  journal = {{SIAM} J. Comput.},
  volume = {53},
  number = {5},
  pages = {1524--1577},
  year = {2024},
  doi = {10.1137/23M1567618},
  timestamp = {Wed, 06 Nov 2024 00:00:00 +0100},
  biburl = {https://dblp.org/rec/journals/siamcomp/KociumakaRRW24.bib},
  bibsource = {dblp computer science bibliography, https://dblp.org},
  _bib2doi_selected = {dblp:/rec/journals/siamcomp/KociumakaRRW24.bib},
  _bib2doi_confirmed = {true},
}

@article{DBLP:journals/tcs/BilleGPRS22,
  author = {Philip Bille and Inge Li G{\o}rtz and Max Rish{\o}j Pedersen and Eva Rotenberg and Teresa Anna Steiner},
  title = {String indexing for top-\emph{k} close consecutive occurrences},
  journal = {Theor. Comput. Sci.},
  volume = {927},
  pages = {133--147},
  year = {2022},
  doi = {10.1016/J.TCS.2022.06.004},
  timestamp = {Sun, 02 Nov 2025 00:00:00 +0100},
  biburl = {https://dblp.org/rec/journals/tcs/BilleGPRS22.bib},
  bibsource = {dblp computer science bibliography, https://dblp.org},
  _bib2doi_selected = {dblp:/rec/journals/tcs/BilleGPRS22.bib},
  _bib2doi_confirmed = {true},
}

@inproceedings{DBLP:conf/stacs/BilleGLPRS24,
  author = {Philip Bille and Inge Li G{\o}rtz and Moshe Lewenstein and Solon P. Pissis and Eva Rotenberg and Teresa Anna Steiner},
  title = {Gapped String Indexing in Subquadratic Space and Sublinear Query Time},
  booktitle = {41st International Symposium on Theoretical Aspects of Computer Science, {STACS} 2024},
  series = {LIPIcs},
  pages = {16:1--16:21},
  publisher = {Schloss Dagstuhl - Leibniz-Zentrum f{\"{u}}r Informatik},
  year = {2024},
  doi = {10.4230/LIPICS.STACS.2024.16},
  timestamp = {Mon, 03 Mar 2025 00:00:00 +0100},
  biburl = {https://dblp.org/rec/conf/stacs/BilleGLPRS24.bib},
  bibsource = {dblp computer science bibliography, https://dblp.org},
  _bib2doi_selected = {dblp:/rec/conf/stacs/BilleGLPRS24.bib},
  _bib2doi_confirmed = {true},
}

@article{DBLP:journals/algorithmica/BilleGPS23,
  author = {Philip Bille and Inge Li G{\o}rtz and Max Rish{\o}j Pedersen and Teresa Anna Steiner},
  title = {Gapped Indexing for Consecutive Occurrences},
  journal = {Algorithmica},
  volume = {85},
  number = {4},
  pages = {879--901},
  year = {2023},
  doi = {10.1007/S00453-022-01051-6},
  timestamp = {Mon, 26 Jun 2023 01:00:00 +0200},
  biburl = {https://dblp.org/rec/journals/algorithmica/BilleGPS23.bib},
  bibsource = {dblp computer science bibliography, https://dblp.org},
  _bib2doi_selected = {dblp:/rec/journals/algorithmica/BilleGPS23.bib},
  _bib2doi_confirmed = {true},
}

@article{DBLP:journals/is/GawrychowskiGSS26,
  author = {Pawel Gawrychowski and Garance Gourdel and Tatiana Starikovskaya and Teresa Anna Steiner},
  title = {Compressed consecutive pattern matching},
  journal = {Inf. Syst.},
  volume = {136},
  pages = {102607},
  year = {2026},
  doi = {10.1016/J.IS.2025.102607},
  timestamp = {Tue, 14 Oct 2025 01:00:00 +0200},
  biburl = {https://dblp.org/rec/journals/is/GawrychowskiGSS26.bib},
  bibsource = {dblp computer science bibliography, https://dblp.org},
  _bib2doi_selected = {dblp:/rec/journals/is/GawrychowskiGSS26.bib},
  _bib2doi_confirmed = {true},
}

@inproceedings{DBLP:conf/cpm/BrodalLPS99,
  author = {Gerth St{\o}lting Brodal and Rune B. Lyngs{\o} and Christian N. S. Pedersen and Jens Stoye},
  title = {Finding Maximal Pairs with Bounded Gap},
  booktitle = {Combinatorial Pattern Matching, 10th Annual Symposium, {CPM} 99},
  series = {Lecture Notes in Computer Science},
  pages = {134--149},
  publisher = {Springer},
  year = {1999},
  url = {https://doi.org/10.1007/3-540-48452-3\_11},
  doi = {10.1007/3-540-48452-3\_11},
  timestamp = {Tue, 21 Mar 2023 00:00:00 +0100},
  biburl = {https://dblp.org/rec/conf/cpm/BrodalLPS99.bib},
  bibsource = {dblp computer science bibliography, https://dblp.org},
  _bib2doi_selected = {dblp:/rec/conf/cpm/BrodalLPS99.bib},
  _bib2doi_confirmed = {true},
}

@inproceedings{DBLP:conf/soda/Muthukrishnan02,
  author = {S. Muthukrishnan},
  title = {Efficient algorithms for document retrieval problems},
  booktitle = {Proceedings of the Thirteenth Annual {ACM-SIAM} Symposium on Discrete Algorithms, {SODA} 2002},
  pages = {657--666},
  publisher = {{ACM/SIAM}},
  year = {2002},
  url = {http://dl.acm.org/citation.cfm?id=545381.545469},
  timestamp = {Mon, 10 May 2021 01:00:00 +0200},
  biburl = {https://dblp.org/rec/conf/soda/Muthukrishnan02.bib},
  bibsource = {dblp computer science bibliography, https://dblp.org},
  _bib2doi_selected = {dblp:/rec/conf/soda/Muthukrishnan02.bib},
  _bib2doi_confirmed = {true},
}

@inproceedings{DBLP:conf/wads/KellerKL07,
  author = {Orgad Keller and Tsvi Kopelowitz and Moshe Lewenstein},
  title = {Range Non-overlapping Indexing and Successive List Indexing},
  booktitle = {Algorithms and Data Structures, 10th International Workshop, {WADS} 2007},
  series = {Lecture Notes in Computer Science},
  pages = {625--636},
  publisher = {Springer},
  year = {2007},
  doi = {10.1007/978-3-540-73951-7\_54},
  timestamp = {Sun, 21 May 2017 01:00:00 +0200},
  biburl = {https://dblp.org/rec/conf/wads/KellerKL07.bib},
  bibsource = {dblp computer science bibliography, https://dblp.org},
  _bib2doi_selected = {dblp:/rec/conf/wads/KellerKL07.bib},
  _bib2doi_confirmed = {true},
}

@article{DBLP:journals/mst/BilleGVV14,
  author = {Philip Bille and Inge Li G{\o}rtz and Hjalte Wedel Vildh{\o}j and S{\o}ren Vind},
  title = {String Indexing for Patterns with Wildcards},
  journal = {Theory Comput. Syst.},
  volume = {55},
  number = {1},
  pages = {41--60},
  year = {2014},
  doi = {10.1007/S00224-013-9498-4},
  timestamp = {Sun, 02 Jun 2019 01:00:00 +0200},
  biburl = {https://dblp.org/rec/journals/mst/BilleGVV14.bib},
  bibsource = {dblp computer science bibliography, https://dblp.org},
  _bib2doi_selected = {dblp:/rec/journals/mst/BilleGVV14.bib},
  _bib2doi_confirmed = {true},
}

@article{DBLP:journals/algorithmica/IliopoulosR09,
  author = {Costas S. Iliopoulos and M. Sohel Rahman},
  title = {Indexing Factors with Gaps},
  journal = {Algorithmica},
  volume = {55},
  number = {1},
  pages = {60--70},
  year = {2009},
  doi = {10.1007/S00453-007-9141-3},
  timestamp = {Thu, 15 Jun 2017 01:00:00 +0200},
  biburl = {https://dblp.org/rec/journals/algorithmica/IliopoulosR09.bib},
  bibsource = {dblp computer science bibliography, https://dblp.org},
  _bib2doi_selected = {dblp:/rec/journals/algorithmica/IliopoulosR09.bib},
  _bib2doi_confirmed = {true},
}

@article{DBLP:journals/algorithmica/BilleG14,
  author = {Philip Bille and Inge Li G{\o}rtz},
  title = {Substring Range Reporting},
  journal = {Algorithmica},
  volume = {69},
  number = {2},
  pages = {384--396},
  year = {2014},
  url = {https://doi.org/10.1007/s00453-012-9733-4},
  doi = {10.1007/s00453-012-9733-4},
  timestamp = {Sun, 02 Jun 2019 01:00:00 +0200},
  biburl = {https://dblp.org/rec/journals/algorithmica/BilleG14.bib},
  bibsource = {dblp computer science bibliography, https://dblp.org},
  _bib2doi_selected = {dblp:/rec/journals/algorithmica/BilleG14.bib},
  _bib2doi_confirmed = {true},
}

@article{DBLP:journals/tcs/NavarroT16,
  author = {Gonzalo Navarro and Sharma V. Thankachan},
  title = {Reporting consecutive substring occurrences under bounded gap constraints},
  journal = {Theor. Comput. Sci.},
  volume = {638},
  pages = {108--111},
  year = {2016},
  doi = {10.1016/J.TCS.2016.02.005},
  timestamp = {Wed, 28 Feb 2024 00:00:00 +0100},
  biburl = {https://dblp.org/rec/journals/tcs/NavarroT16.bib},
  bibsource = {dblp computer science bibliography, https://dblp.org},
  _bib2doi_selected = {dblp:/rec/journals/tcs/NavarroT16.bib},
  _bib2doi_confirmed = {true},
}

@article{DBLP:journals/algorithmica/Charalampopoulos21,
  author = {Panagiotis Charalampopoulos and Tomasz Kociumaka and Manal Mohamed and Jakub Radoszewski and Wojciech Rytter and Tomasz Walen},
  title = {Internal Dictionary Matching},
  journal = {Algorithmica},
  volume = {83},
  number = {7},
  pages = {2142--2169},
  year = {2021},
  doi = {10.1007/S00453-021-00821-Y},
  timestamp = {Tue, 13 Jul 2021 01:00:00 +0200},
  biburl = {https://dblp.org/rec/journals/algorithmica/Charalampopoulos21.bib},
  bibsource = {dblp computer science bibliography, https://dblp.org},
  _bib2doi_selected = {dblp:/rec/journals/algorithmica/Charalampopoulos21.bib},
  _bib2doi_confirmed = {true},
}

@book{DBLP:books/cu/Gusfield1997,
  author = {Dan Gusfield},
  title = {Algorithms on Strings, Trees, and Sequences - Computer Science and Computational Biology},
  publisher = {Cambridge University Press},
  year = {1997},
  doi = {10.1017/CBO9780511574931},
  isbn = {0-521-58519-8},
  timestamp = {Mon, 29 Jul 2019 01:00:00 +0200},
  biburl = {https://dblp.org/rec/books/cu/Gusfield1997.bib},
  bibsource = {dblp computer science bibliography, https://dblp.org},
  _bib2doi_selected = {dblp:/rec/books/cu/Gusfield1997.bib},
  _bib2doi_confirmed = {true},
}

@book{DBLP:books/aw/Baeza-YatesR2011,
  author = {Ricardo Baeza{-}Yates and Berthier A. Ribeiro{-}Neto},
  title = {Modern Information Retrieval - the concepts and technology behind search, Second edition},
  publisher = {Pearson Education Ltd., Harlow, England},
  year = {2011},
  url = {http://www.mir2ed.org/},
  isbn = {978-0-321-41691-9},
  timestamp = {Sun, 22 Sep 2019 01:00:00 +0200},
  biburl = {https://dblp.org/rec/books/aw/Baeza-YatesR2011.bib},
  bibsource = {dblp computer science bibliography, https://dblp.org},
  _bib2doi_selected = {dblp:/rec/books/aw/Baeza-YatesR2011.bib},
  _bib2doi_confirmed = {true},
}

@inproceedings{DBLP:conf/isaac/BrodalFGL09,
  author = {Gerth St{\o}lting Brodal and Rolf Fagerberg and Mark Greve and Alejandro L{\'{o}}pez{-}Ortiz},
  title = {Online Sorted Range Reporting},
  booktitle = {Algorithms and Computation, 20th International Symposium, {ISAAC} 2009},
  series = {Lecture Notes in Computer Science},
  volume = {5878},
  pages = {173--182},
  publisher = {Springer},
  year = {2009},
  url = {https://doi.org/10.1007/978-3-642-10631-6\_19},
  doi = {10.1007/978-3-642-10631-6\_19},
  timestamp = {Sat, 19 Oct 2019 01:00:00 +0200},
  biburl = {https://dblp.org/rec/conf/isaac/BrodalFGL09.bib},
  bibsource = {dblp computer science bibliography, https://dblp.org},
  _bib2doi_selected = {dblp:/rec/conf/isaac/BrodalFGL09.bib},
  _bib2doi_confirmed = {true},
}

@article{DBLP:journals/ipl/JuanLW09,
  author = {M. T. Juan and Jia Jie Liu and Yue{-}Li Wang},
  title = {Errata for "Faster index for property matching"},
  journal = {Inf. Process. Lett.},
  volume = {109},
  number = {18},
  pages = {1027--1029},
  year = {2009},
  doi = {10.1016/J.IPL.2009.06.009},
  timestamp = {Sun, 01 Feb 2026 00:00:00 +0100},
  biburl = {https://dblp.org/rec/journals/ipl/JuanLW09.bib},
  bibsource = {dblp computer science bibliography, https://dblp.org},
  _bib2doi_selected = {dblp:/rec/journals/ipl/JuanLW09.bib},
  _bib2doi_confirmed = {true},
}

@article{DBLP:journals/ipl/IliopoulosR08,
  author = {Costas S. Iliopoulos and M. Sohel Rahman},
  title = {Faster index for property matching},
  journal = {Inf. Process. Lett.},
  volume = {105},
  number = {6},
  pages = {218--223},
  year = {2008},
  doi = {10.1016/j.ipl.2007.09.004},
  timestamp = {Wed, 14 Jun 2017 01:00:00 +0200},
  biburl = {https://dblp.org/rec/journals/ipl/IliopoulosR08.bib},
  bibsource = {dblp computer science bibliography, https://dblp.org},
  _bib2doi_selected = {dblp:/rec/journals/ipl/IliopoulosR08.bib},
  _bib2doi_confirmed = {true},
}

@inproceedings{DBLP:conf/icalp/Ruzic08,
  author = {Milan Ruzic},
  title = {Constructing Efficient Dictionaries in Close to Sorting Time},
  booktitle = {Automata, Languages and Programming, 35th International Colloquium, {ICALP} 2008},
  series = {Lecture Notes in Computer Science},
  volume = {5125},
  pages = {84--95},
  publisher = {Springer},
  year = {2008},
  doi = {10.1007/978-3-540-70575-8\_8},
  timestamp = {Tue, 23 May 2017 01:00:00 +0200},
  biburl = {https://dblp.org/rec/conf/icalp/Ruzic08.bib},
  bibsource = {dblp computer science bibliography, https://dblp.org},
  _bib2doi_selected = {dblp:/rec/conf/icalp/Ruzic08.bib},
  _bib2doi_confirmed = {true},
}

@article{DBLP:journals/siamcomp/ManberM93,
  author = {Udi Manber and Gene Myers},
  title = {Suffix Arrays: {A} New Method for On-Line String Searches},
  journal = {{SIAM} J. Comput.},
  volume = {22},
  number = {5},
  pages = {935--948},
  year = {1993},
  doi = {10.1137/0222058},
  timestamp = {Wed, 14 Nov 2018 00:00:00 +0100},
  biburl = {https://dblp.org/rec/journals/siamcomp/ManberM93.bib},
  bibsource = {dblp computer science bibliography, https://dblp.org},
  _bib2doi_selected = {dblp:/rec/journals/siamcomp/ManberM93.bib},
  _bib2doi_confirmed = {true},
}

@incollection{DBLP:books/daglib/GonnetBS92,
  author = {Gaston H. Gonnet and Ricardo A. Baeza{-}Yates and Tim Snider},
  title = {New Indices for Text: {Pat} Trees and {Pat} Arrays},
  booktitle = {Information Retrieval: Data Structures and Algorithms},
  publisher = {Prentice-Hall},
  pages = {66--82},
  year = {1992},
  timestamp = {Tue, 06 Aug 2019 01:00:00 +0200},
  biburl = {https://dblp.org/rec/books/ph/frakesB92/GonnetBS92.bib},
  bibsource = {dblp computer science bibliography, https://dblp.org},
  _bib2doi_selected = {dblp:/rec/books/ph/frakesB92/GonnetBS92.bib},
  _bib2doi_confirmed = {true},
}

@article{DBLP:journals/siamcomp/GrossiV05,
  author = {Roberto Grossi and Jeffrey Scott Vitter},
  title = {Compressed Suffix Arrays and Suffix Trees with Applications to Text Indexing and String Matching},
  journal = {{SIAM} J. Comput.},
  volume = {35},
  number = {2},
  pages = {378--407},
  year = {2005},
  doi = {10.1137/S0097539702402354},
  timestamp = {Tue, 21 Mar 2023 00:00:00 +0100},
  biburl = {https://dblp.org/rec/journals/siamcomp/GrossiV05.bib},
  bibsource = {dblp computer science bibliography, https://dblp.org},
  _bib2doi_selected = {dblp:/rec/journals/siamcomp/GrossiV05.bib},
  _bib2doi_confirmed = {true},
}

@article{DBLP:journals/jacm/FerraginaM05,
  author = {Paolo Ferragina and Giovanni Manzini},
  title = {Indexing Compressed Text},
  journal = {J. {ACM}},
  volume = {52},
  number = {4},
  pages = {552--581},
  year = {2005},
  doi = {10.1145/1082036.1082039},
  timestamp = {Sun, 19 Jan 2025 00:00:00 +0100},
  biburl = {https://dblp.org/rec/journals/jacm/FerraginaM05.bib},
  bibsource = {dblp computer science bibliography, https://dblp.org},
  _bib2doi_selected = {dblp:/rec/journals/jacm/FerraginaM05.bib},
  _bib2doi_confirmed = {true},
}

@inproceedings{DBLP:conf/soda/GagieNP18,
  author = {Travis Gagie and Gonzalo Navarro and Nicola Prezza},
  title = {Optimal-Time Text Indexing in {BWT}-Runs Bounded Space},
  booktitle = {Proceedings of the Twenty-Ninth Annual {ACM-SIAM} Symposium on Discrete Algorithms ({SODA})},
  pages = {1459--1477},
  publisher = {{SIAM}},
  year = {2018},
  doi = {10.1137/1.9781611975031.96},
  timestamp = {Wed, 28 Feb 2024 00:00:00 +0100},
  biburl = {https://dblp.org/rec/conf/soda/GagieNP18.bib},
  bibsource = {dblp computer science bibliography, https://dblp.org},
  _bib2doi_selected = {dblp:/rec/conf/soda/GagieNP18.bib},
  _bib2doi_confirmed = {true},
}

@article{DBLP:journals/csur/Navarro21,
  author = {Gonzalo Navarro},
  title = {Indexing Highly Repetitive String Collections, Part {I}: Repetitiveness Measures},
  journal = {{ACM} Comput. Surv.},
  volume = {54},
  number = {2},
  pages = {29:1--29:31},
  year = {2021},
  doi = {10.1145/3434399},
  timestamp = {Wed, 28 Feb 2024 00:00:00 +0100},
  biburl = {https://dblp.org/rec/journals/csur/Navarro21a.bib},
  bibsource = {dblp computer science bibliography, https://dblp.org},
  _bib2doi_selected = {dblp:/rec/journals/csur/Navarro21a.bib},
  _bib2doi_confirmed = {true},
}

@inproceedings{DBLP:conf/focs/Farach97,
  author = {Martin Farach},
  title = {Optimal Suffix Tree Construction with Large Alphabets},
  booktitle = {38th Annual Symposium on Foundations of Computer Science, {FOCS} 1997},
  pages = {137--143},
  publisher = {{IEEE} Computer Society},
  year = {1997},
  url = {https://doi.org/10.1109/SFCS.1997.646102},
  doi = {10.1109/SFCS.1997.646102},
  timestamp = {Thu, 23 Mar 2023 00:00:00 +0100},
  biburl = {https://dblp.org/rec/conf/focs/Farach97.bib},
  bibsource = {dblp computer science bibliography, https://dblp.org},
  _bib2doi_selected = {dblp:/rec/conf/focs/Farach97.bib},
  _bib2doi_confirmed = {true},
}

@inproceedings{DBLP:conf/focs/Weiner73,
  author = {Peter Weiner},
  title = {Linear Pattern Matching Algorithms},
  booktitle = {14th Annual Symposium on Switching and Automata Theory ({SWAT} 1973)},
  pages = {1--11},
  publisher = {{IEEE} Computer Society},
  year = {1973},
  doi = {10.1109/SWAT.1973.13},
  timestamp = {Thu, 23 Mar 2023 00:00:00 +0100},
  biburl = {https://dblp.org/rec/conf/focs/Weiner73.bib},
  bibsource = {dblp computer science bibliography, https://dblp.org},
  _bib2doi_selected = {dblp:/rec/conf/focs/Weiner73.bib},
  _bib2doi_confirmed = {true},
}

@inproceedings{DBLP:conf/sdm/0001CCGGLPP24,
  author = {Giulia Bernardini and Huiping Chen and Alessio Conte and Roberto Grossi and Veronica Guerrini and Grigorios Loukides and Nadia Pisanti and Solon P. Pissis},
  title = {Utility-Oriented String Mining},
  booktitle = {Proceedings of the 2024 {SIAM} International Conference on Data Mining, {SDM} 2024},
  pages = {190--198},
  publisher = {{SIAM}},
  year = {2024},
  doi = {10.1137/1.9781611978032.22},
  timestamp = {Wed, 29 Oct 2025 00:00:00 +0100},
  biburl = {https://dblp.org/rec/conf/sdm/0001CCGGLPP24.bib},
  bibsource = {dblp computer science bibliography, https://dblp.org},
  _bib2doi_selected = {dblp:/rec/conf/sdm/0001CCGGLPP24.bib},
  _bib2doi_confirmed = {true},
}

@article{DBLP:journals/tcs/AmirCIKZ08,
  author = {Amihood Amir and Eran Chencinski and Costas S. Iliopoulos and Tsvi Kopelowitz and Hui Zhang},
  title = {Property matching and weighted matching},
  journal = {Theor. Comput. Sci.},
  volume = {395},
  number = {2-3},
  pages = {298--310},
  year = {2008},
  doi = {10.1016/J.TCS.2008.01.006},
  timestamp = {Wed, 17 Feb 2021 00:00:00 +0100},
  biburl = {https://dblp.org/rec/journals/tcs/AmirCIKZ08.bib},
  bibsource = {dblp computer science bibliography, https://dblp.org},
  _bib2doi_selected = {dblp:/rec/journals/tcs/AmirCIKZ08.bib},
  _bib2doi_confirmed = {true},
}

@article{DBLP:journals/talg/KociumakaKRRW20,
  author = {Tomasz Kociumaka and Marcin Kubica and Jakub Radoszewski and Wojciech Rytter and Tomasz Walen},
  title = {A Linear-Time Algorithm for Seeds Computation},
  journal = {{ACM} Trans. Algorithms},
  volume = {16},
  number = {2},
  pages = {27:1--27:23},
  year = {2020},
  doi = {10.1145/3386369},
  timestamp = {Mon, 07 Apr 2025 01:00:00 +0200},
  biburl = {https://dblp.org/rec/journals/talg/KociumakaKRRW20.bib},
  bibsource = {dblp computer science bibliography, https://dblp.org},
  _bib2doi_selected = {dblp:/rec/journals/talg/KociumakaKRRW20.bib},
  _bib2doi_confirmed = {true},
}

@inproceedings{DBLP:conf/esa/Radoszewski23,
  author = {Jakub Radoszewski},
  title = {Linear Time Construction of Cover Suffix Tree and Applications},
  booktitle = {31st Annual European Symposium on Algorithms, {ESA} 2023},
  series = {LIPIcs},
  volume = {274},
  pages = {89:1--89:17},
  publisher = {Schloss Dagstuhl - Leibniz-Zentrum f{\"{u}}r Informatik},
  year = {2023},
  doi = {10.4230/LIPICS.ESA.2023.89},
  timestamp = {Wed, 30 Aug 2023 01:00:00 +0200},
  biburl = {https://dblp.org/rec/conf/esa/Radoszewski23.bib},
  bibsource = {dblp computer science bibliography, https://dblp.org},
  _bib2doi_selected = {dblp:/rec/conf/esa/Radoszewski23.bib},
  _bib2doi_confirmed = {true},
}

@book{DBLP:books/daglib/0020103,
  author = {Maxime Crochemore and Christophe Hancart and Thierry Lecroq},
  title = {Algorithms on strings},
  publisher = {Cambridge University Press},
  year = {2007},
  isbn = {978-0-521-84899-2},
  timestamp = {Wed, 23 Mar 2011 00:00:00 +0100},
  biburl = {https://dblp.org/rec/books/daglib/0020103.bib},
  bibsource = {dblp computer science bibliography, https://dblp.org},
  _bib2doi_selected = {dblp:/rec/books/daglib/0020103.bib},
  _bib2doi_confirmed = {true},
}

@article{DBLP:journals/dam/Regnier00,
  author = {Mireille R{\'{e}}gnier},
  title = {A unified approach to word occurrence probabilities},
  journal = {Discret. Appl. Math.},
  volume = {104},
  number = {1-3},
  pages = {259--280},
  year = {2000},
  url = {https://doi.org/10.1016/S0166-218X(00)00195-5},
  doi = {10.1016/S0166-218X(00)00195-5},
  timestamp = {Thu, 11 Feb 2021 00:00:00 +0100},
  biburl = {https://dblp.org/rec/journals/dam/Regnier00.bib},
  bibsource = {dblp computer science bibliography, https://dblp.org},
  _bib2doi_selected = {dblp:/rec/journals/dam/Regnier00.bib},
  _bib2doi_confirmed = {true},
}

@article{DBLP:journals/tkde/GanLFCTY21,
  author = {Wensheng Gan and Jerry Chun{-}Wei Lin and Philippe Fournier{-}Viger and Han{-}Chieh Chao and Vincent S. Tseng and Philip S. Yu},
  title = {A Survey of Utility-Oriented Pattern Mining},
  journal = {{IEEE} Trans. Knowl. Data Eng.},
  volume = {33},
  number = {4},
  pages = {1306--1327},
  year = {2021},
  url = {https://doi.org/10.1109/TKDE.2019.2942594},
  doi = {10.1109/TKDE.2019.2942594},
  timestamp = {Mon, 03 Mar 2025 00:00:00 +0100},
  biburl = {https://dblp.org/rec/journals/tkde/GanLFCTY21.bib},
  bibsource = {dblp computer science bibliography, https://dblp.org},
  _bib2doi_selected = {dblp:/rec/journals/tkde/GanLFCTY21.bib},
  _bib2doi_confirmed = {true},
}

@inproceedings{DBLP:conf/icalp/BrodalLOP02,
  author = {Gerth St{\o}lting Brodal and Rune B. Lyngs{\o} and Anna {\"{O}}stlin and Christian N. S. Pedersen},
  title = {Solving the String Statistics Problem in Time ${O}(n \log n)$},
  booktitle = {Automata, Languages and Programming, 29th International Colloquium, {ICALP} 2002},
  series = {Lecture Notes in Computer Science},
  pages = {728--739},
  publisher = {Springer},
  year = {2002},
  doi = {10.1007/3-540-45465-9\_62},
  timestamp = {Sat, 19 Oct 2019 01:00:00 +0200},
  biburl = {https://dblp.org/rec/conf/icalp/BrodalLOP02.bib},
  bibsource = {dblp computer science bibliography, https://dblp.org},
  _bib2doi_selected = {dblp:/rec/conf/icalp/BrodalLOP02.bib},
  _bib2doi_confirmed = {true},
}

@article{DBLP:journals/algorithmica/ApostolicoP96,
  author = {Alberto Apostolico and Franco P. Preparata},
  title = {Data Structures and Algorithms for the String Statistics Problem},
  journal = {Algorithmica},
  volume = {15},
  number = {5},
  pages = {481--494},
  year = {1996},
  url = {https://doi.org/10.1007/BF01955046},
  doi = {10.1007/BF01955046},
  timestamp = {Wed, 17 May 2017 01:00:00 +0200},
  biburl = {https://dblp.org/rec/journals/algorithmica/ApostolicoP96.bib},
  bibsource = {dblp computer science bibliography, https://dblp.org},
  _bib2doi_selected = {dblp:/rec/journals/algorithmica/ApostolicoP96.bib},
  _bib2doi_confirmed = {true},
}

\end{document}